\documentclass[twoside,11pt]{article}

\usepackage{blindtext}

\usepackage[preprint]{jmlr2e}

\usepackage{algorithm2e}
\usepackage{amsmath}
\usepackage{bm}
\usepackage{bbm}
\usepackage{enumitem}
\usepackage{centernot}

\DeclareMathOperator*{\argmin}{arg\,min}

\DeclareMathOperator*{\esssup}{ess\,sup}

\newcommand{\norm}[1]{\left\lVert#1\right\rVert}
\newcommand{\floor}[1]{\left\lfloor#1\right\rfloor}
\newcommand\independent{\protect\mathpalette{\protect\independenT}{\perp}}
\def\independenT#1#2{\mathrel{\rlap{$#1#2$}\mkern2mu{#1#2}}}
\newcommand{\ind}[1]{\mathbbm{1}_{\{#1\}}}

\DeclareMathOperator{\Var}{Var}

\DeclareMathOperator{\E}{\mathbb{E}}
\DeclareMathOperator{\p}{\mathbb{P}}

\newcommand{\mcX}{\mathcal{X}}
\newcommand{\mcY}{\mathcal{Y}}
\newcommand{\mcZ}{\mathcal{Z}}

\newcommand{\bbR}{\mathbb{R}}

\usepackage{lastpage}

\ShortHeadings{Formal Privacy for Partially Private Data}{Seeman, Reimherr, and Slavkovic}
\firstpageno{1}

\begin{document}

\title{Formal Privacy for Partially Private Data}

\author{\name Jeremy Seeman \email jhs5496@psu.edu \\
       \addr Department of Statistics\\
       Pennsylvania State University \\
       University Park, PA 16802
       \AND
       \name Matthew Reimherr \email mlr36@psu.edu \\
       \addr Department of Statistics\\
       Pennsylvania State University \\
       University Park, PA 16802 
       \AND 
       \name Aleksandra Slavkovic \email abs12@psu.edu \\
       \addr Department of Statistics\\
       Pennsylvania State University \\
       University Park, PA 16802}


\maketitle

\begin{abstract}
Differential privacy (DP) quantifies privacy loss by analyzing noise injected into output statistics. For non-trivial statistics, this noise is necessary to ensure finite privacy loss. However, data curators frequently release collections of statistics where some use DP mechanisms and others are released as-is, i.e., without additional randomized noise. Consequently, DP alone cannot characterize the privacy loss attributable to the entire collection of releases. In this paper, we present a privacy formalism, $(\epsilon, \{ \Theta_z\}_{z \in \mathcal{Z}})$-Pufferfish ($\epsilon$-TP for short when $\{ \Theta_z\}_{z \in \mathcal{Z}}$ is implied), a collection of Pufferfish mechanisms indexed by realizations of a random variable $Z$ representing public information not protected with DP noise. First, we prove that this definition has similar properties to DP. Next, we introduce mechanisms for releasing partially private data (PPD) satisfying $\epsilon$-TP and prove their desirable properties. We provide algorithms for sampling from the posterior of a parameter given PPD. We then compare this inference approach to the alternative where noisy statistics are deterministically combined with Z. We derive mild conditions under which using our algorithms offers both theoretical and computational improvements over this more common approach. Finally, we demonstrate all the effects above on a case study on COVID-19 data.
\end{abstract}


\section{Introduction}

Differential privacy (DP) \citep{dwork2006calibrating, dwork2014algorithmic}, a framework for releasing statistical outputs with provable privacy guarantees, has earned its popularity through its desirable properties such as robustness to post-processing and methodological transparency. DP offers semantic interpretations about what adversaries can learn in excess of their prior information, via Pufferfish \citep{kifer2014pufferfish}, or by comparing information gains counterfactually based on whether an individual's data was used via coupled worlds \cite{bassily2013coupled} or posterior-to-posterior semantics \citep{kasiviswanathan2008note}. These semantics describe privacy loss attributable to the DP mechanism, regardless of other arbitrary prior information. In recent years, there has been a massive proliferation of DP methods, with notable developments in private selection problems \citep{mcsherry2007mechanism,reimherr2019kng,asi2020}, statistical inference tasks \citep{awan2018differentially,canonne2019structure,biswas2020coinpress}, and synthetic data generation \citep{snoke2018pmse,torkzadehmahani2019dp,mckenna2018optimizing}. 

In the central trust model, a centralized data curator processes all statistics published using data from a single confidential database. DP can only guarantee finite total privacy losses for a collection of outputs in the setting where each non-trivial output is released via a DP mechanism. However, data curators often release a collection of statistics where some are published with DP mechanisms and others are not. In these settings, DP can describe the privacy loss attributable to the outputs released with DP mechanisms, but this fails to non-trivially describe the privacy guarantees of everything the data curator actually released. Our central goal in this paper is to recover some formal privacy guarantees when DP does not meaningfully describe the entirety of the data curator's actions due to the presence of what we call ``public information," i.e., statistics derived from the confidential database available without privacy-preserving randomized noise needed to satisfy DP. We investigate the unique problems associated with partially private data (PPD), data releases which contain both public information and query responses sanitized for formal privacy preservation.

There are many different circumstances in which PPD can arise; we briefly list some important cases. First, data curators may be legally or contractually bound to release public information along with their private releases. For tabular count data, public summary statistics may be released as-is, such as mariginal summaries \citep{slavkovic2010synthetic,Abowd2019} or survey methodology descriptions in data sharing agreements \citep{seeman2021posterior}. Similarly, data curators may make database schema and modeling decisions non-privately, such as choosing tuning parameters or setting privacy budgets based on particular utility goals \citep{xiao2021optimizing}. Second, information about the database sample may be public prior to implementing formally private methodology. Most data curators do not currently use DP methods, posing unresolved challenges in reconciling future releases with past ones when results are temporally autocorrelated \citep{quick2021generating}. These problems affect countless data curators, and it poses a significant barrier to wider adoption of privacy-preserving methodology when DP fails to describe what data curators are actually releasing in practice. 

\subsection{Contributions}
\label{sec:formalism-intro-contributions}

In this paper, we present a privacy formalism, $(\epsilon, \{ \Theta_z\}_{z \in \mathcal{Z}})$-Pufferfish ($\epsilon$-TP for short when $\{ \Theta_z\}_{z \in \mathcal{Z}}$ is implied), which is a collection of Pufferfish \citep{kifer2014pufferfish} mechanisms indexed by realizations of a random variable $Z$. This variable represents public information about the confidential data not protected with DP noise. Recall that Pufferfish considers comparisons not of adjacent databases (as in DP), but conditional distributions given information about the confidential database; our framework considers these same distributions, but additionally conditioned on the realization of public information $Z = z$ that we happen to observe.

Our contributions are as follows. 1) We prove that $\epsilon$-TP maintains similar properties and semantic interpretations as DP (Lemmas \ref{lem:etp-bayes} through \ref{lem:etp-parcomp}). 2) We propose two release mechanisms for PPD, the Wasserstein Exponential mechanism (Theorem \ref{thm:wassexp} generalizing \citep{mcsherry2007mechanism}) and Wasserstein $K$-norm gradient mechanisms (Theorem \ref{thm:wasskng} generalizing \citep{reimherr2019kng}), and prove their useful properties. 3) we provide algorithms for sampling from the posterior of a parameter given PPD, based on rejection and importance sampling (Algorithms \ref{alg:ppi-rej} and \ref{alg:ppi-is}, respectively). 4) We then compare this inference approach to the common alternative where noisy statistics are deterministicly combined with public information. We then derive conditions where our mechanisms require adding asymptotically less noise (Corollary \ref{cor:manifold_ext_vs_intr}) and offer stochastic improvements in inferential quality (Theorem \ref{thm:expfamdom}). 5) Finally, we demonstrate all the effects above on a case study on COVID-19 data (Section \ref{sec:data-analysis}). 

{\bf Broader Impacts:} Organizations hoping to share formally private results are frequently required to release results in a non-DP manner about the same database, even if the exact correspondence is not deterministic (for example, statistics based on a random sample from the database). We argue that this case warrants special attention because unlike the examples given in \citet{dwork2017exposed}, which correctly states that  population-level inferences about individuals cannot be quantified as privacy violations, here the data curator is directly in charge of multiple releases from the same database, only some of which may satisfy DP. This is a problem for which DP alone offers no substantive solutions other than to narrow the scope of which privacy losses are quantified. In this paper, we argue against this approach, as it fails to hold data curators accountable for additional privacy risks emanating from non-DP releases. We may hope that in the future, organizations are better equipped to implement DP as intended. Until such a time when organizations' data processing requirements can accommodate DP, our work provides a solution to a real and present formal privacy problem.

\subsection{Related literature} 

Research in DP, synonymous with \citep{dwork2006calibrating}, has spurned hundreds of similar definitions \citep{desfontaines2019sok}. Some of those (e.g., Pufferfish privacy \citep{kifer2014pufferfish}, Blowfish privacy \citep{he2014blowfish}, coupled worlds privacy \citep{bassily2013coupled}, and dependent differential privacy \citep{liu2016dependence}) can be extended to accommodate certain kinds of public information. We also note that extensions of Pufferfish have been used to address other conceptual issues with DP, such as accounting for non-individual database attribute disclosure risks \citep{zhang2022attribute}.

Our framework differs in a few key ways. First, these previous frameworks offer guarantees that do not change contextually based on the realized value of the public information. This is why we cannot naively apply Pufferfish to our problem, because the space of candidate conditional distributions for our database changes based on the realized value of $Z$; we need not consider a ``worst-case" public information realization among all possible values of $Z$ when the realization is known publicly. Second, by treating public information as a random variable, we consider the effects of probabilistic public information by smoothly interpolating between the best case scenario (in which $X$ and $Z$ are entirely independent) and the worst case scenario (in which $X$ and $Z$ are perfectly correlated). 

Existing works on formal privacy in the presence of public information typically focus on congeniality, or agreement between results derived from private and public information \citep{barak2007privacy,hay2010boosting,ding2011differentially,abowd2019census}. Attempts to reinterpret these guarantees with semantic similarity to DP have been studied in \citep{ashmead2019effective,gong2020congenial}. Our work does not require congeniality, but we prove that our proposed mechanisms can satisfy it. Recent work has also considered mechanism design in the presence of additional public information \citep{bassily2020private} and provides helpful theoretical bounds, but it requires that the private and public samples are independently and identically distributed (iid). We relax the iid assumption, since our interest is in changes to mechanism design and valid statistical inference explicitly due to dependencies between private and public information.

\section{Privacy guarantees of joint private and public releases}

We provide all complete proofs in Appendix 
\ref{apx:proofs}. Throughout this paper, we consider a parameter of interest $\theta \in \Theta$, confidential database $X \in \mathcal{X}^n$, public information $Z \in \mathcal{Z}$, and sanitized private release $Y \in \mathcal{Y}$. In some cases, we consider $Y^* = h(Y, Z)$, where $Y^*$ post-processes both $Y$ and $Z$ according to the arbitrary function $h$. This is the approach taken by the Geometric mechanism \citep{Ghosh2012}, the Top-Down algorithm \citep{Abowd2019}, and private hypothesis testing \citep{canonne2019structure}, among many other methods. 

\subsection{Motivation: $\epsilon$-DP semantics with public information}
\label{sec:post-break}

In this section, we motivate why we want to characterize the privacy loss attributable to both $Y$ and $Z$ simultaneously. We first review $\epsilon$-DP in its original formulation. Note that in this Subsection \ref{sec:post-break}, for notational ease, we temporarily assume that the distributions in question admit mass functions. 

\begin{definition}[$\epsilon$-DP \citep{dwork2006calibrating,dwork2006our}]
Let $x \in \mathcal{X}^n$ and let $x' \in \mathcal{X}^n$ be any database created by modifying one entry in $x$. A randomized algorithm $\{ M(x): x \in \mathcal{X}^n \}$, i.e., a collection of probability distributions over $\mathcal{Y}$ indexed by elements of $\mathcal{X}^n$, satisfies $\epsilon$-DP iff for all such $x,x'$ and $y \in \mathcal{Y}$.  
\begin{equation}
\label{eq:edp_orig}
\p(M(x) = y) \leq \p(M(x') = y) e^\epsilon .
\end{equation}
\end{definition}

As presented in Equation \ref{eq:edp_orig}, DP is a property of a collection of marginal probability distributions, which does not account for randomness in the underlying data $X$. In \cite{kasiviswanathan2014semantics}, the authors describe the ``posterior-to-posterior" interpretation of $\epsilon$-DP that provides privacy protections accounting for priors on $\mathcal{X}$. Specifically, if $\pi(\cdot)$ is an arbitrary prior mass function on $\mathcal{X}$, the authors define:
\begin{equation}
\label{eq:post-game}
\pi_{i,v}(x \mid y) \triangleq \frac{\p(M(x_{i,v}) = y) \pi(x)}{\sum_{x^* \in \mathcal{X}^n} \p(M(x^*_{i,v}) = y) \pi(x^*)} .
\end{equation}
Above, $x_{i,v}$ is the database $x$ where the $i$th record is replaced with arbitrary, data-independent value $v$. The authors show, for all $i \in [n]$, $v \in \mathcal{X}$, $x \in \mathcal{X}^n$, and any possible prior $\pi$:
\begin{equation}
d_{\mathrm{TV}}\left( \pi(\cdot \mid y), \pi_{i,v}(\cdot \mid y) \right) \leq e^\epsilon - 1,
\label{eq:post-semantic}
\end{equation}
where $d_\mathrm{TV}$ is total variational distance. 

One might conjecture that posterior-to-posterior semantics ought to cover public information, since it places no restrictions on the priors $\pi(\cdot)$. This is true in terms of what privacy losses are attributable to the DP mechanism itself when the mechanism does not incorporate information from $Z$. Specifically, suppose the equivalent of Equation \ref{eq:post-game} including public information $Z$ takes the form:
\begin{equation}
\label{eq:post-game-z-ind}
\pi_{i,v}(x \mid y, z) \triangleq \frac{\p(M(x_{i,v}) = y) \pi(x \mid z)}{\sum_{x^* \in \mathcal{X}^n} \p(M(x^*_{i,v}) = y) \pi(x^* \mid z)} .
\end{equation}
Then, using the same argument as \cite{kasiviswanathan2014semantics}, we have:
\begin{equation}
d_{\mathrm{TV}}\left( \pi(\cdot \mid y, z), \pi_{i,v}(\cdot \mid y, z) \right) \leq e^\epsilon - 1.
\label{eq:post-semantic-z-ind}
\end{equation}
Our concern in this paper is \emph{not} Equation \ref{eq:post-game-z-ind}; instead, we analyze the scenario when the mechanism form $M$ \emph{does} depend on the public information $Z$, in which case the equivalent expression to Equation \ref{eq:post-game-z-ind} could be ill-defined as a conditional distribution if we are conditioning on sets of measure zero. In this case, Equation \ref{eq:post-semantic-z-ind} need not hold.
Moreover, as a conditional distribution, the distribution of the mechanism output $Y \mid X, Z$ need not be the same as $Y \mid X$, as we have seen in \cite{Seeman2020} and \cite{gong2020congenial}. We will also see this throughout the paper, particularly in the case study (Section \ref{sec:data-analysis}). This captures the case when public information is used to decide \emph{how} to implement the mechanism, such as choosing its form or setting hyperparameters. Therefore if we want to consider the formal privacy guarantees from releasing both $Y$ and $Z$, we cannot rely on DP's semantic interpretations alone. 

\subsection{Proposal \texorpdfstring{$\epsilon$}{epsilon}-TP}
\label{sec:etp}

Next, we motivate why we consider $Y \mid X, Z$ (as opposed to $Y, Z \mid X$) and we propose our definition. To move towards this inferential perspective, we now review Pufferfish privacy \citep{kifer2014pufferfish}:

\begin{definition}[Pufferfish \citep{kifer2014pufferfish}]
\label{def:puff}Let $\mathbb{D}$ be a collection of probability distributions on $(\mathcal{X}^n, \mathcal{F}_X)$ indexed by a parameter $\theta \in \Theta$ (called ``data evolution scenarios" \footnote{more commonly referred to as ``data generating distributions" in statistics}). Let $\mathbb{S} \subset \mathcal{F}_X$ be a collection of events (called ``secrets"), and let $\mathbb{S}_{\mathrm{pairs}} \subset \mathbb{S} \times \mathbb{S}$ be a collection of pairs of disjoint events in $\mathbb{S}$ (called ``discriminative pairs"). Then a mechanism that releases random variable $Y$ in $(\mathcal{Y}, \mathcal{F}_Y)$ satisfies $\epsilon$-Pufferfish$(\mathbb{S}, \mathbb{S}_{\mathrm{pairs}}, \mathbb{D})$ privacy if for all $B \in \mathcal{F}_Y$, $\theta \in \mathbb{D}$, and $(s_1, s_2) \in \mathbb{S}_{\mathrm{pairs}}$ such that $\p(s_i \mid \theta) \notin \{ 0, 1 \}$ for $i = 1, 2$:
\begin{equation}
\label{eq:puff}
\begin{cases}
\p(Y \in B \mid s_1, \theta) \leq e^\epsilon \p(Y \in B \mid s_2, \theta) \\
\p(Y \in B \mid s_2, \theta) \leq e^\epsilon \p(Y \in B \mid s_1, \theta).
\end{cases}
\end{equation}
\end{definition}

Note that we will use $\mathbb{D}_{\mathrm{DP}}$, $\mathbb{S}_{\mathrm{DP}}$, and $\mathbb{S}_{\mathrm{pairs,DP}}$ to refer to the Pufferfish components that represent a semantic interpretation of $\epsilon$-DP's guarantees in the language of Pufferfish. In particular, let $\mathcal{H} \triangleq \{ h_i \}_{i=1}^N$ be a population of $N$ individuals and $\mathcal{R} = \{ x_i \}_{i=1}^n$ be a sample of records. Define the events:
\begin{equation}
\begin{cases}
\sigma_i &\triangleq \text{``record $r_i$ belongs to individual $h_i$ in the data"} \\
\sigma_{i,x} &\triangleq \text{``record $r_i$ belongs to individual $h_i$ in the data and has value $x \in \mathcal{X}$"}
\end{cases}.
\end{equation}
Then:
\begin{equation}
\begin{cases}
\mathbb{S}_{\mathrm{DP}} \triangleq \{ \sigma_{i,x} \mid i \in [N], x \in \mathcal{X} \} \cup \{ \sigma_{i}^c \mid i \in [N] \} \\
\mathbb{S}_{\mathrm{pairs,DP}} \triangleq \{ (\sigma_{i,x}, \sigma_i^c) \mid i \in [N], x \in \mathcal{X} \}
\end{cases}.
\end{equation}
The corresponding data model depends on $\pi_i \triangleq \p(h_i \in \mathcal{R})$ and independent densities $f_i(r_i)$ on $(\mathcal{X}, \mathcal{F}_X)$, yielding $\mathbb{D}_{\mathrm{DP}}$ as the set of all distributions with densities given by:
\begin{equation}
\mathbb{D}_{\mathrm{DP}} \triangleq \left\{ \p_X \mid f_X(\bm{x}) = \prod_{r_i \in \mathcal{R}} \pi_i f_i(r_i) \prod_{r_i \notin \mathcal{R}} (1 - \pi_i)  \right\} .
\end{equation}
As expected, \citep[Thm 6.1]{kifer2014pufferfish} showed that $\epsilon$-DP implies $\epsilon$-Pufferfish$(\mathbb{S}_{\mathrm{DP}}, \mathbb{S}_{\mathrm{pairs,DP}}, \mathbb{D}_{\mathrm{DP}})$. 

Considering that $Z$ is a public released statistic, we may intuitively want to analyze $Y, Z \mid X$, since we are ultimately interested in formalizing the privacy loss of $Y$ and $Z$ jointly. Doing so, unfortunately, raises many of the same problems that motivated DP in the first place. The joint distribution $Y, Z \mid X$ depends both on $Z \mid X$ and $Y \mid X, Z$. However, $Z \mid X$ can be a degenerate distribution when $Z$ is a direct function of $X$, say $Z = g(X)$. In cases like these, given $s_1, s_2 \in \mathbb{S}_{\mathrm{pairs,DP}}$, $g(X) \mid X, s_1$ or $g(X) \mid X, s_2$ may not be well-defined mathematically as conditional distributions, since we could be conditioning on measure zero sets. Therefore we cannot generically satisfy Equation \ref{eq:puff} simply by inserting $Z$ into our statistical outputs. Instead, staying true to the spirit of DP, we want to account for privacy loss due to the release of $Y$ \emph{when the way} $Y$ \emph{is released depends on} $Z$. By focusing on this unit of analysis, we are implicitly considering a temporal ordering where $Z$ is released first, and $Y$ is released according to a mechanism which depends on $Z$. By doing so, we explicitly account for how $Z$ \emph{informs} how DP mechanisms are implemented, such as based on past data releases, fitness-for-use goals, or any other mechanism setting depending on the realized confidential database.  

Here, we finally introduce our new privacy formalism. Our goal with this approach is to extend Pufferfish to accommodate $Z$ while maintaining similar desirable properties as $\epsilon$-DP. Within this setup, our target data evolution scenarios are those conditioned on the existing public results $Z$. This yields the intuition for our definition:

\begin{definition}[$\epsilon$-TP]\label{def:etp} Let $Z$ be a random variable on $(\mathcal{Z}, \mathcal{F}_Z)$ that depends on $X$. For each $z \in \mathcal{Z}$, let $\mathbb{D}_{\mathrm{DPz}}$ be a collection of conditional distributions for $X \mid Z = z$ indexed by $\theta_z \in \Theta_z$. We say that the mechanism that releases a random variable $Y$ satisfies $(\epsilon, \{\Theta_z \}_{z \in \mathcal{Z}})$-Pufferfish (which we will call $\epsilon$-TP for short when $\{\Theta_z \}_{z \in \mathcal{Z}}$ is implied) if for all $z \in \mathcal{Z}$ and $B \in \mathcal{F}_Y$, for all distributions $\theta_z \in \Theta_z$, and for all $(s_1, s_2) \in \mathbb{S}_{\mathrm{pairs,DPz}}$, where
\begin{equation}
\label{eq:puff_secrets_w_uncertainty}
\mathbb{S}_{\mathrm{pairs,DPz}} \triangleq \{ (s_1, s_2) \in \mathbb{S}_{\mathrm{pairs,DP}} \mid \p(s_i \mid \theta_z) \notin \{ 0, 1 \} \quad \forall \ i \in \{1, 2\}, \theta_z \in \mathrm{D}_{\mathrm{DPz}} \},
\end{equation}
we have
\begin{equation}
\label{eq:puffcond}
\begin{cases}
\p(Y \in B \mid s_1, \theta_z) \leq e^\epsilon \p(Y \in B \mid s_2, \theta_z) \\
\p(Y \in B \mid s_2, \theta_z) \leq e^\epsilon \p(Y \in B \mid s_1, \theta_z).
\end{cases}
\end{equation}
\end{definition}

Note that this defines a \textit{collection} of Pufferfish mechanisms, each of which is induced by a particular realization of the public information $z \in \mathcal{Z}$. To do this, our framework requires that the user specify a class of conditional distributions. This poses an important design problem: as the space of possible conditional distributions $\Theta_z$ becomes larger, the framework provides its guarantees across increasingly worse dependencies between $X$ and $Z$. However, if this space becomes too large, it becomes impossible for few if any mechanisms to satisfy this definition. This trade-off re-frames the ``no free lunch" problem of \citep{kifer2011no}, referenced in the original Pufferfish paper, strictly in terms of what dependence we allow on public information $Z$ (in Section \ref{apx:delta_z}, we demonstrate this trade-off in practice for our case study).

Next, we outline the properties of $\epsilon$-TP, highlighting where additional assumptions need to be made about relationships between $X$, $Z$, and $Y$:

\begin{lemma}[Bayesian semantics of $\epsilon$-TP]
\label{lem:etp-bayes}In the setup for Definition \ref{def:etp}, Equation \ref{eq:puffcond} equivalently implies, for all $B \in \mathcal{F}_Y$, $z \in \mathcal{Z}$, and associated $\theta \in \mathbb{D}_{\mathrm{DPz}}$ and $(s_1, s_2) \in \mathbb{S}_{\mathrm{pairs,DPz}}$:
\begin{equation}
\label{eq:puffbayescond}e^{-\epsilon} \leq \frac{\p(s_1 \mid Y \in B, \theta_z) / \p(s_2 \mid  Y \in B, \theta_z)}{\p(s_1 \mid \theta_z) / \p(s_2 \mid \theta_z)} \leq e^\epsilon.
\end{equation}
\end{lemma}

\begin{lemma}[$\epsilon$-TP post-processing]
\label{lem:etp-pp}Let $h: \mathcal{Y} \times \mathcal{Z} \mapsto \mathcal{Y}^*$ be a measurable function of $Y$ and $Z$ \footnote{For the purposes of this paper, we refer to this as a ``post-processing function." However, because our results are not DP, we do not mean this in the same sense as $\epsilon$-DP post-processing.}. Then releasing $Y^* \triangleq h(Y,Z)$ is $\epsilon$-TP.
\end{lemma}

\begin{lemma}[Sequential composition for $\epsilon$-TP]
\label{lem:etp-seqcomp}Suppose $Y_1$ is $\epsilon_1$-TP, $Y_2$ is $\epsilon_2$-TP, and $Y_1 \independent Y_2 \mid X, Z$ (i.e. $Y_1$ and $Y_2$ are conditionally independent given $X$ and $Z$). Then the random vector $(Y_1, Y_2)$ is $(\epsilon_1+\epsilon_2)$-TP.
\end{lemma}

\begin{lemma}[Parallel composition for $\epsilon$-TP]
\label{lem:etp-parcomp}Consider a disjoint partition $X = \bigcup_{b \in B} X_b$ where $Y_1, \dots Y_B$ each satisfy $\epsilon$-TP and $Y_b \mid X, Z \sim Y_b \mid X_b, Z$ for all $b \in [B]$. Then $(Y_1, \dots, Y_B)$ is $\epsilon$-TP.
\end{lemma}

\subsection{Mechanisms that satisfy \texorpdfstring{$\epsilon$}{epsilon}-TP}

At first glance, $\epsilon$-TP may seem substantially harder to use than $\epsilon$-DP. In this section, we will alleviate those fears by discussing how existing tools from $\epsilon$-DP can be easily repurposed into $\epsilon$-TP. First, we consider the case when the set of plausible databases for $\mathcal{X}^n$ is restricted by public information. This captures many common forms of public information; examples include public tabular summaries and information gleaned from timing attacks (where only a subset of possible databases induce mechanisms that have a particular observed run-time).

\begin{lemma}
\label{lem:etp-cond}In the absence of any public information, let $Y$ be an $\epsilon$-DP release over $\mathcal{Y} = \mathcal{X}^n$., i.e. the marginal distribution of $Y$ is that of a randomized algorithm satisfying $\epsilon$-DP. Next, consider $Z \triangleq \ind{X \in \mathcal{X}^{*n}}$. Then $Y \mid Z = 1$, i.e. $Y \mid Y \in \mathcal{X}^{*n}$, is $(2\epsilon, \Theta_{\mathrm{DPz}})$-TP where:
$$
\Theta_{\mathrm{DPz}} = \{ f(x) \propto f^*(x) \ind{x \in \mathcal{X}^*} \mid f^* \in \mathbb{D}_{\mathrm{DP}} \}
$$
In words, $\mathbb{D}_{\mathrm{DPz}}$ is the same as $\mathbb{D}_{\mathrm{DP}}$ with each density renormalized over only those databases conforming to $Z$.
\end{lemma}
                 
Next, we need to define a generic class of mechanisms that satisfies $\epsilon$-TP. In one dimension given some output function $h: \mathcal{Y} \mapsto \mathbb{R}$, \citep{song2017pufferfish} defines a sensitivity analogue:
$$
\Delta \triangleq \sup_{s_1, s_2 \in \mathbb{S}_{\mathrm{pairs}}} \sup_{\theta \in \Theta} \ \ W_\infty \left( \p(h(\cdot) \mid s_1, \theta), \p(h(\cdot) \mid s_2, \theta) \right), 
$$
where $W_\infty$ is the Wasserstein-$\infty$ metric:
$$
W_\infty(\mu, \nu) \triangleq \inf_{\gamma \in \Gamma(\mu, \nu)} \mathrm{ess sup}_{\gamma} \norm{\mu - \nu}_1 . 
$$
In the above definition, $\mu$ and $\nu$ are measures on $(\mathcal{Y}, \mathcal{F}_Y)$, and $\gamma$ is the set of all joint distributions on $(\mathcal{Y} \times \mathcal{Y}, \mathcal{F}_Y \times \mathcal{F}_Y)$. In Theorem \ref{thm:wassexp}, we propose a generalization of the Wasserstein mechanism \citep{song2017pufferfish} for arbitrary loss functions that relaxes the restrictions on the output space and extends $\mathcal{X}$ to be an arbitrary metric space.

\begin{theorem}[Wasserstein Exponential Mechanism]
\label{thm:wassexp}
Let $(\mathcal{X}, d)$ be a metric space and $L_x: \mathcal{X} \times \mathcal{Y} \mapsto [0, \infty]$ be a loss function. Fix $z \in \mathcal{Z}$. Define:
$$
\Delta_z \triangleq \sup_{\theta_z \in \Theta_z} \sup_{(s_1, s_2)\in \mathbb{S}_{\mathrm{pairs, DPz}} } W_{\infty} \left(\p(\cdot \mid \theta_z, s_1), \p(\cdot \mid \theta_z, s_2) \right)
$$ 
and 
$$
\sigma(\Delta_z) = \sup \left\{ |L_x(y) - L_{x'}(y)| \mid x,x' \in \mathcal{X}^n, d(x, x') \leq \Delta_z  \right\}.
$$
Then releasing a sample $Y$ with density given by
$$
f_X(y) \propto \exp\left(-\frac{\epsilon L_X(y)}{2\sigma(\Delta_z)} \right),
$$
with respect to base measure $\nu_Z$ which may depend on $Z$, satisfies $\epsilon$-TP.
\end{theorem}
\begin{proof}
(Sketch) for the optimal transport solution $\gamma^*$ which achieves $\Delta_Z$, $\sigma(\Delta_Z)$ bounds how much the loss function can change in a ball of radius $\Delta_Z$ around any pair of two conditional distributions given both the public $Z$ and one of the two secrets $s_1$ and $s_2$, respectively. This then satisfies $\epsilon$-TP by construction for any choice of $\nu_Z$.
\end{proof}

If we want to place additional restrictions on our loss function and public information, we can ensure that the errors introduced due to privacy are asymptotically negligible, i.e.,  $O_P(n^{-1})$, relative to the errors from sampling, i.e., $O_P(n^{-1/2})$. This can be seen by extending the $K$-norm gradient mechanism \citep{reimherr2019kng} to our setting. 

\begin{theorem}[Wasserstein K-norm Gradient Mechanism] 
\label{thm:wasskng}In the same setting as Theorem \ref{thm:wassexp}, suppose $\mathcal{Y} \subseteq \mathbb{R}^d$ is convex, $\norm{\cdot}_K$ is a $K$-norm, and there exists a function $\sigma_y(\Delta_z): \mathcal{Y} \mapsto \mathbb{R}^+$ where, whenever $d(X, X') \leq \Delta_Z$, $\norm{\nabla L_X(y) - \nabla L_{X'}(y)}_K \leq \sigma_y(\Delta_z)$. Then releasing a sample $Y$ with density given by:
$$
f_X(y) \propto \exp\left(-\frac{\epsilon \norm{\nabla L_X(y)}_K }{2\sigma_y(\Delta_z)} \right),
$$
with respect to base measure $\nu_Z$ which may depend on $Z$, satisfies $\epsilon$-TP. 
\end{theorem}

\begin{theorem}[Asymptotic Utility of Wasserstein K-norm Gradient Mechanism]
\label{thm:wasskngutil} In the setting of \citep[Theorem 3.2]{reimherr2019kng}, under some regularity conditions (listed in the supplementary materials), if $\{(X_n, Y_n, Z_n)\}_{n=1}^\infty$ is a sequence of random variables in the setting of Theorem \ref{thm:wasskng} and $Y_n^* = \argmin_{y \in \mathcal{Y}} L_{X_n}(y)$ where $Y_n^* \to_p y^* \in \mathcal{Y}$ and $\Delta_{Z_n} \geq \Delta$ with probability 1 for all $n \in \mathbb{N}$, then:
$$
\norm{Y_n - Y_n^*}_K = O_P(n^{-1}).
$$
\end{theorem}
\begin{proof}
(Sketch) similar to \citep{reimherr2019kng}, we Taylor expand the target density and use Scheff\'es theorem to establish the asymptotic distribution. By requiring that $Z$ supports the true optimization solution with high probability, but $Z$ remains less informative for $Y$ than $X$, $Y$ has an asymptotic $K$-norm density over the asymptotic base measure.
\end{proof}

In both mechanisms, we can interpret our new sensitivity $\sigma(\Delta_z)$ relative to the pure DP sensivitiy of the loss function. Specifically, if the collection of distributions $\mathbb{D}_{\mathrm{DPz}}$ is sufficiently rich for each $z \in \mathcal{Z}$, then in general the mechanisms defined in Theorems \ref{thm:wassexp} and \ref{thm:wasskng} induces a sensitivity at least as large as that in the $\epsilon$-DP analogue. We formalize this in Lemma \ref{lem:etp-rep}:

\begin{lemma}
\label{lem:etp-rep}Define $\mathbb{D}_X \triangleq \{ \xi_{\bm{x}} \mid \bm{x} \in \mathcal{X}^n \}$. Suppose for every distribution $\theta_{\bm{x}} \in \mathbb{D}_X$, there exists a sequence of distributions $\{ \theta_{z, \bm{x}}^{(m)} \}_{m=1}^\infty \subseteq \mathbb{D}_{\mathrm{DPz}}$ and a pair of secrets $(s_1, s_2) \in \mathbb{S}_{\mathrm{pairs,DPz}}$ such that:
$$
\theta_{z, \bm{x}}^{(m)} \to_D \xi_{\bm{x}} \quad \forall \bm{x} \in \mathcal{X}^n, \quad \quad \p(\cdot \mid s_1, \theta_{z, \bm{x}}), \p(\cdot \mid s_2, \theta_{z, \bm{x}}) \text{ well defined.}
$$
Then:
$$
\sup_{x,x'; d_H(x,x') = 1} \sup_{y \in \mathcal{Y}} \norm{L_X(y) - L_{X'}(y)} \leq \sigma(\Delta_z).
$$
\end{lemma}

\subsection{Effect of Public Information on $\epsilon$-TP Mechanism Design}

We present two examples showing how different choices of distributions on $Z$ affect implementation of the Wasserstein exponential mechanism (complete derivation details in Appendix \ref{apx:ex_derivs}). Example \ref{ex:mvn_cond} demonstrates how different distribution choices for $X \mid Z$ affect the sensitivity of our mechanism. Example \ref{ex:count_sdl} discusses how our framework can provide both $\epsilon$-TP and statistical disclosure limitation (SDL) guarantees simultaneously.

\begin{example}[Multivariate Conditional Normal Mean] 
\label{ex:mvn_cond}
Suppose our goal is to privately estimate a sample mean $\mu$. We assume a priori that for our database, $X_1, \dots, X_n \in [-\Delta/2, \Delta / 2]$ so that $\sup_{\bm{X} \sim \bm{X}'} | \overline{X} - \overline{X'} | \leq \frac{\Delta}{n}$. Then we can sample $Y$ from the standard exponential mechanism satisfying $\epsilon$-DP from the density with $L_1$ loss:
\begin{equation}
\label{eq:mean_expmech_dp}
f_x(y) \propto \exp\left(-\frac{n \epsilon}{2\Delta} |y - \overline{x}| \right) \ind{y \in [-\Delta/2, \Delta/2]}.
\end{equation}

Alternatively, we cast this problem in terms of $\epsilon$-TP. Let 
\begin{equation}
\label{eq:mean_expmech_puff}
\begin{cases}
\mathbb{D} \triangleq \{ X_1, \dots, X_n \overset{\text{iid}}{\sim} N(\mu, \sigma^2) \mid \mu \in [-\Delta / 2, \Delta / 2], \sigma \in [0, \infty) \} \\
\mathbb{S}_{\mathrm{pairs}} \triangleq \{ \{ \omega \in \Omega \mid X_i(\omega) = v\} \mid i \in [n], v \in [-\Delta / 2, \Delta / 2] \}.
\end{cases}
\end{equation}
Suppose we want to estimate $\overline{X}$ and we observe the following joint distribution of $\overline{X}, X_1, Z$
\begin{equation}
\label{eq:mean_expmech_joint}
\begin{pmatrix}
\overline{X} \\
\begin{pmatrix}
X_1 \\
Z
\end{pmatrix}
\end{pmatrix} \sim N\left( \begin{pmatrix}
\mu \\ \begin{pmatrix}
\mu \\ 
\mu_Z
\end{pmatrix}
\end{pmatrix}, 
\begin{pmatrix}
\frac{\sigma^2}{n} & \Sigma_{XVZ}^T \\
\Sigma_{XVZ} & \Sigma_{VZ}
\end{pmatrix}
\right),
\end{equation}
then following the calculations in Appendix \ref{apx:ex_derivs}, Example \ref{apx:mvn_cond}, the Wasserstein exponential mechanism sensitivity is given by

\begin{equation}
\label{eq:mean_expmech_wass}
\sup_{\theta \in \mathbb{D}} \sup_{(s_1, s_2) \in \mathbb{S}_{\mathrm{pairs}}} W_\infty(\p_{\theta,v_1}, \p_{\theta,v_2}) = \sup_{\theta \in \mathbb{D}} \left[ \Sigma_{XVZ}^T \Sigma_{VZ}^{-1} \begin{pmatrix}
\Delta \\ z - \mu_{Z}
\end{pmatrix} \right] . 
\end{equation}
Therefore, the sensitivity of the Wasserstein exponential mechanism loss depends explicitly on what kinds of dependence we allow between the private and public information, i.e., how we choose $\mathbb{D}$.
\end{example}

Remarks on Example \ref{ex:mvn_cond}: first, if we allow for arbitrary dependence, i.e., any $\Sigma_{XVZ}$, then the expression in Equation \ref{eq:mean_expmech_wass} can be infinite. This can be seen as an example of the ``no free lunch problem," as arbitrary dependence offers potentially no privacy in the worst-case scenario \citep{kifer2011no}. Second, note that when our choice of sensitivity, $\Delta$, is part of both the data generating process and the secret pairs, we leave ourselves open to potential privacy model misspecification (i.e., how $X$ and $Z$ are related). This is explicitly true here, but such a phenomena was only implicitly true for when interpreting $\epsilon$-DP's semantics using Pufferfish as shown in \citep{kifer2014pufferfish}.

\begin{example}[Count Data Satisfying $\epsilon$-TP and SDL measures]
\label{ex:count_sdl}
Suppose we are interested in binary data $X_i \in \{0, 1\}$ for $i \in [n]$ and a single counting query $S_n = \sum_{i=1}^n X_i$. The secret pairs are,
\begin{equation}
\label{eq:ex_count_sdl_sec_main}
    s_{ij} = \{ \omega \in \Omega \mid X_i(\omega) = j \}, \quad \mathbb{S}_{\mathrm{pairs}} \triangleq \left\{ (s_{i0}, s_{i1}) \mid i \in [n] \right\}.
\end{equation}
We consider public information $Z$ to be the property that releasing $(S_n, n - S_n)$ satisfies different Statistical Disclosure Limitation (SDL) measures, namely $k$-anonymity \citep{sweeney2002k} and $L_1$-distance $t$-closeness \citep{li2007t} to a public statistic $t^*$:
\begin{equation}
\label{eq:ex_count_sdl_main}
\begin{cases}
S_n \text{ satisfies $k$-anonymity} \implies S_n \in \{ k, k + 1, \dots, n - k \} \\
S_n \text{ satisfies $L_1$-$t$-closeness} \implies \left| \frac{S_n}{n} - t^* \right| \leq t .
\end{cases}
\end{equation}
In the full details (in Appendix \ref{apx:ex_derivs}, Example \ref{apx:count_sdl}), we show how satisfying SDL only restricts secret pairs and not the sensitivity of the Wasserstein Exponential Mechanism. This means $\epsilon$-TP only confers protections for those databases which plausibly agree with $Z$, but this partial schema could be constrained by the practical privacy guarantees afforded by SDL.
\end{example}

\section{Inference using PPD}

In this section, we consider estimating a parameter $\theta$ given both $Y$ and $Z$. Most papers in the DP literature are concerned with \textit{designing} both an optimal release mechanism \textit{and} an optimal estimator associated with that mechanism. However, because PPD is often used to produce synthetic data or other general-purpose statistics, inference typically requires \textit{adjusting} inferences based on the outputs from a fixed release mechanism, implying not all outputs can provide optimal estimators \citep{slavkovic2022statistical}. We will see these dynamics play out here as well.

\subsection{General inferential utility and congeniality}
\label{sec:inf_utility}

For the purposes of this paper, we informally define congeniality as the property that a sanitized statistic $Y$ agrees with public information. Formally, we say $Y$ and $Z$ are congenial for functions $g_1: \mathcal{X} \mapsto \mathbb{R}^d$ and $g_2: \mathcal{Z} \mapsto \mathbb{R}^d$ iff, for any output set $B \in \mathcal{F}_Y$, confidential database $x \in \mathcal{X}$, and realized public information $z \in \mathcal{Z}$, we have $g_1(Y) = g_2(Z)$ with probability 1. For example, if $Y$ is a collection of noisy sum queries over disjoint subsets of records, and $Z$ is the non-private total sum, we might reasonably expect the sum of entries in $Y$ to equal $Z$, i.e. the marginal sum of $Y$ is congenial with the public sum $Z$. 

Throughout this section, we will consider the effect of post-processing functions $Y^* = h(Y, Z)$. In many cases, $Y^*$ is the solution to an optimization problem where we find the nearest $y' \in \mathcal{Y}$ such that $y$ and $z$ are congenial, i.e.:
$$
h(y, z) = \mathrm{Proj}(y, \mathcal{Z}) \triangleq \argmin_{y' \in \mathcal{Y}} d(y, y') \text{ s.t. } g_1(y) = g_2(z)
$$

In implementing the Wasserstein exponential mechanism, we do not require any assumptions on the base measure $\nu$, so we don't require congeniality between $Y$ and $Z$. However, independent of whether $\nu$ encodes any prior belief about $Y$, we already saw that any particular $z \in Z$ can restrict which $X$s are plausible, with some works already addressing this problem \citep{gong2020congenial,gao2021subspace,soto2021differential}. In particular, the manifold approach of Theorem 4 in \citep{soto2021differential} demonstrates the benefit of  and loss function for $\epsilon$-TP:

\begin{corollary}[Asymptotic Errors for Congenial Manifolds]
\label{cor:manifold_ext_vs_intr}
Let $\mcY \triangleq \bbR^{d_1}$ be the output space for an instance of the Wasserstein exponential mechanism, and let $\mcZ \triangleq \mathrm{Span}(A)$ for some matrix $A$ of rank $d_2 < d_1$ (this may correspond to the results of summation queries on the confidential data, for example). Let $Y_1$ and $Y_2$ be two $\epsilon$-TP releases estimating a mean $\overline{X}$ using Laplace noise, supported on $\bbR^{d_1}$ and $\mathrm{Span}(A)$, respectively. Then under the regularity conditions of \citep[Theorem 4]{soto2021differential}:
$$
\begin{cases}
\E[\norm{\mathrm{Proj}(Y_1, \mathcal{Z}) - \overline{X}}_2^2] = O\left(\frac{d_1^2}{n^2 \epsilon^2} \right) \\
\E[\norm{Y_2 - \overline{X}}_2^2] = O\left(\frac{d_1 d_2}{n^2 \epsilon^2} \right).
\end{cases}
$$
\end{corollary}

These results can be extended to nonlinear manifolds under the regularity assumptions in \citep{soto2021differential} due to the uniqueness of the local linear approximation as $n \to \infty$. 

While the above result impacts the degree of errors added due to privacy, we see similar differences when 
We can use Blackwell's classical comparison of experiments theorem \citep{blackwell1953equivalent}, which trivially yields the following lemma: 

\begin{lemma}[Expected performance of post-processed Bayes estimators]
For any Bayesian decision problem estimating some  $\tau(\theta)$ with loss function $L: \Theta \times \Theta \mapsto [0, \infty)$, the Bayes estimators based on $Z,Y$ and $Y^*$, $\delta_{Z,Y}$ and $\delta_{Y^*}$ respectively, are related by:
$$
\E[L(\delta_{Z,Y}, \tau(\theta))] \leq \E[L(\delta_{Y^*}, \tau(\theta))].
$$
In particular, choosing $L(\theta, \delta) = (\theta - \delta)^2$ and $\tau(\theta) = \theta$ implies
$$
\E[\Var(\theta \mid Z,Y)] \leq \E[\Var(\theta \mid Y^*)].
$$
\end{lemma}

These results are statements about the estimator performance in expectation. For a broader class of problems, inference can be also improved in a way that stochastically dominates using $Y^*$:

\begin{theorem}[Exponential family inference]
\label{thm:expfamdom} Let $X_1, \dots, X_n \overset{\text{iid}}{\sim} f_\theta$ with density:
$$
f_\theta(x) = h(x) \exp\left(\eta(\theta) T(x) - A(\eta(\theta)) \right).
$$
Let $Y$ be an instance of the Wasserstein exponential mechanism \citep{mcsherry2007mechanism} with density: 
$$
g_x(y) \propto \exp\left(-\frac{\epsilon}{2\sigma(\Delta_z)} L(\norm{y - T(x)}) \right) 
$$
for some loss function $L$ that depends only on a norm $\norm{y - T(x)}$. Let $Z = h(X)$ and let $Y^* = \mathrm{Proj}(Y, \mathcal{Z})$, so that $Y, Y^* \in \mathcal{Y}$. Then $Y$ has a monotone likelihood ratio in $\theta$. Furthermore, define the test: 
$$
H_0: \theta \leq \theta_0, \quad \quad H_1: \theta > \theta_0.
$$
For any unbiased test $\phi: \mathcal{Y}^* \mapsto [0, 1]$ for $\theta$ based on $Y^*$, there exists a uniformly more powerful test $\phi': \mathcal{Y} \times \mathcal{Z} \mapsto [0, 1]$. If $\p_\theta(Y \neq Y^*) > 0$ for all $\theta \in \Theta$, then this improvement is strict.
\end{theorem}
\begin{proof}
(Sketch) we use the main theorem in \citep{wijsman1985useful} to establish the monotone likelihood ratio. For the second result, note that if multiple $Y,Z$ values are mapped to the same $Y*$, then one can construct a more powerful test by calculating finer-grained and more efficient Type I and II errors.
\end{proof}

Note that for Theorem \ref{thm:expfamdom}, it is often the case that as $n \to \infty$, $\norm{Y_n^* - Y_n} \to_P 0$; in words, the effect of post-processing for congeniality with $Z$ may be asymptotically negligible. This indicates the importance of finite-sample analyses for these problems, as asymptotic analyses may fail to detect such a difference.

\subsection{Posterior inference from the target distribution}

In the previous section, we motivated the construction of estimators directly using $Y$ and $Z$ instead of relying on a post-processed $Y^*$. Because the sampling distribution of estimators from two potentially dependent sources can be difficult to analytically derive, we turn to Bayesian computational approaches and propose Algorithm \ref{alg:ppi-rej} that relies on rejection sampling to produce exact draws from $\theta \mid Y, Z$, and Algorithm \ref{alg:ppi-is} that obtains valid inference based on importance sampling; complete algorithm are available in Appendix \ref{apx:algs} and inference validity proofs are in Appendix \ref{apx:proofs}.

Algorithm \ref{alg:ppi-rej} is an exact sampling algorithm for $\theta \mid Y, Z$, an extension of the algorithm for exact inference from a private posterior from \citep{gong2019exact} now conditioning on $Z$ as used empirically in \cite{Seeman2020}. In \cite{Fearnhead2012}, the authors demonstrate an important connection between approximate Bayesian computation (ABC) and inference under measurement error. They noted that samples drawn from an approximate posterior distribution based on an observed summary statistic can be alternatively interpreted as exact samples from a posterior distribution conditional on a summary statistic measured with noise. \cite{gong2019exact} showed that this agrees with the sufficient statistic perturbation setup for private Bayesian inference (e.g., \cite{foulds2016theory}). 

\begin{theorem}
\label{thm:ppi-rej-valid}Algorithm \ref{alg:ppi-rej} samples from $\theta \mid Y, Z$
\end{theorem}
\begin{proof}
(sketch) this rejection algorithm follows from extending \citep{Fearnhead2012,gong2019exact} with additional conditioning for $Z$ in each proposal stage. Note that while Algorithm \ref{alg:ppi-rej} applies to the Wasserstein exponential mechanism releases, it could generically apply to any privacy mechanism satisfying the analytic tractability requirements outlined in \citep{gong2020transparent}.
\end{proof}

\begin{theorem}
\label{thm:ppi-is-valid}Algorithm \ref{alg:ppi-is} estimates $\widehat{E}_m[\theta \mid Y, Z],$ and $\widehat{E}_m[\theta \mid Y, Z] \to_P \E[a(\theta) \mid Y, Z]$ as $m \to \infty$.
\end{theorem}
\begin{proof}
(sketch) analogous to the proof for Theorem \ref{thm:ppi-rej-valid}. Note this method is more computationally efficient if one is only interested in a Bayes estimator for $\theta \mid Y, Z$.
\end{proof}

\section{Data analysis example}
\label{sec:data-analysis}

Here we present a sample data analysis to demonstrate the impact of public information on synthesis strategies and downstream inferences. Our analysis is based on data from PA Department of Health \citep{pahealth} and the U.S. Census Bureau's Current Population Survey (CPS) \citep{ipums}. Both data sets are published publicly and available for academic public use under their respective terms of service (available in the references); therefore, our analysis poses no additional privacy concerns nor data use issues. We synthesize sanitized monthly county-level COVID-19 case and death counts, $Y$, where we assume the data curator chose to release PPD at a fixed point in time when previous results had no formal privacy guarantees. Therefore our public information, $Z$, is the total number of COVID-19 cases at the current month, and county-level COVID-19 cases at the previous month, mirroring a realistic problem faced by organizations releasing information about the same population at regular time intervals. 

\subsection{Data Description}

The Pennsylvania Department of Health collects data on the total number of reported COVID-19 cases and deaths per county \citep{pahealth}. Data from the US Census Bureau's Current Population Survey (CPS) provided estimated population counts by demographic strata, accessed through IPUMS \citep{ipums}. For the purposes of our analysis, we will treat these values as fixed population totals (i.e., we temporarily ignore other errors due to sampling and weighting schemes). 

\subsection{Synthesis Methods and Base Measure Choice}

We will examine the effect of public information on our ability to perform inference from these data sources. Throughout this section, we use the following notation, with $[n] \triangleq \{ 1, 2, \dots, n \}$:
$$
\begin{cases}
\displaystyle j \in \{ 1, \dots, J\} &\triangleq \text{Pennsylvania counties, } J = 67 \\
\displaystyle t \in \{ 1, \dots, T \} &\triangleq \text{Year-month periods, } T=24 \\
\displaystyle X_{j,t}^{(c)} &\triangleq \text{Number of COVID-19 cases in county $j$ at time $t$} \\
\displaystyle X_{j,t}^{(d)} &\triangleq \text{Number of COVID-19 deaths in county $j$ at time $t$} \\
\end{cases}
$$
Our goal will be to privatize $X_{j,t}^{(c)}$ and $X_{j,t}^{(d)}$ under the following public information scenarios:
$$
Z_{t} = \begin{cases}
\displaystyle X_{j,t-1}^{(c)} = x_{j,t-1}^{(c)} \\
\displaystyle X_{j,t-1}^{(d)} = x_{j,t-1}^{(d)} \\
\displaystyle \sum_{j=1}^J X_{j,t} = s_{j,t}^{(c)} \\
\displaystyle \p(X_{j,t}^{(c)} \geq X_{j,t}^{(d)}) = 1
\end{cases}
$$
For this analysis, we will consider the following synthesis procedures for case counts:

\begin{enumerate}
    \item \textbf{Post-processing}: first, we synthesize:
    $$
    \begin{cases}
    Y_{j, t}^{(c)} = X_{j, t}^{(c)} + \varepsilon_{j, t}^{(c)} \\
    Y_{j, t}^{(d)} = X_{j, t}^{(d)} + \varepsilon_{j, t}^{(d)}
    \end{cases} \quad \qquad \varepsilon_{j, t}^{(c)}, \varepsilon_{j, t}^{(d)} \overset{\text{iid}}{\sim} \mathrm{DiscreteLaplace}\left( \frac{\epsilon}{\Delta} \right)
    $$
    Then, we will perform deterministic two-stage post-processing:
    $$
    \begin{pmatrix}
    \tilde{\bm{Y}}_t^{(c)} \\
    \tilde{\bm{Y}}_t^{(d)}
    \end{pmatrix} = \argmin_{\bm{y} \in \mathbb{R}^{2J}} \norm{\begin{pmatrix} \bm{y}^{(c)} \\ \bm{y}^{(d)} \end{pmatrix} - \begin{pmatrix}
    \bm{Y}_t^{(c)} \\
    \bm{Y}_t^{(d)}
    \end{pmatrix}}_{L_2}^2 \text{s.t.}
    $$
    $$
    \begin{cases}
    \displaystyle \sum_{j=1}^J y_j^{(c)} = s_{j,t} \\ 
    \displaystyle y_j^{(c)} \geq y_j^{(d)} \geq 0 \quad \forall j \in [J]
    \end{cases}
    $$
    Next:
    $$
    \begin{pmatrix}
    \bm{Y}_t^{(c)*} \\
    \bm{Y}_t^{(d)*}
    \end{pmatrix}
    = 
    \begin{pmatrix}
    \floor{\tilde{\bm{Y}}_t^{(c)}} \\
    \floor{\tilde{\bm{Y}}_t^{(d)}}
    \end{pmatrix} + \argmin_{\bm{y} \in \{0, 1\}^{2J}}
    \norm{\begin{pmatrix} \bm{y}^{(c)} \\ \bm{y}^{(d)} \end{pmatrix} - \begin{pmatrix}
    \tilde{\bm{Y}}_t^{(c)} \\
    \tilde{\bm{Y}}_t^{(d)}
    \end{pmatrix}}_{L_1} \text{s.t.} 
    $$
    $$
    \begin{cases}
    \displaystyle \sum_{j=1}^J y_j^{(c)} = s_{j,t} - \floor{\tilde{\bm{Y}}_t^{(c)}} \\ 
    \displaystyle y_j^{(c)} + \floor{\tilde{Y}_t^{(c)}} \geq y_j^{(d)} \floor{\tilde{Y}_t^{(d)}} \geq 0 \quad \forall j \in [J]
    \end{cases}
    $$
    
    \item \textbf{Wasserstein Mechanism}
    Alternatively, we'll compare this method to the following version of the Wasserstein exponential mechanism with loss function:
    $$
    L_{\bm{X}}(\bm{Y}) = \norm{\begin{pmatrix} \bm{y}^{(c)} \\ \bm{y}^{(d)} \end{pmatrix} - \begin{pmatrix}
    \bm{Y}_t^{(c)} \\
    \bm{Y}_t^{(d)}
    \end{pmatrix}}_{L_2}
    $$
    Note that for this section, we will assume $\Delta_Z = 2$ across mechanisms, without specifying the choice of conditional distributions which leads to this value. To understand how this choice affects $\Delta_Z$, please see Section \ref{apx:delta_z}.
    
    We'll use the public information to inform the base measure in three different secnarios:
    \begin{enumerate}
        \item Naive base measure:
        $$
        \nu_{\bm{Z}}(\bm{Y}) \propto \ind{\bm{Y} \in \mathbb{Z}^{J}}
        $$
        \item Deterministic congenial base measure: 
        $$
        \nu_{\bm{Z}}(\bm{Y}) \propto \ind{\bm{Y} \in \mathbb{Z}^{J}, \sum_{i=1}^{J} Y_{j,t}^{(c)} = s_t^c , Y_{j,t}^{(c)} \geq Y_{j,t}^{(d)} \geq 0 }
        $$
        \item Prior congenial base measure:
        $$
        \nu_{\bm{Z}}(\bm{Y}) \propto \phi(\bm{Y}; s_t^{(c)}, \bm{X}_{t-1}^{(c)}) \ind{\bm{Y} \in \mathbb{Z}^{J}, \sum_{i=1}^{J} Y_{j,t}^{(c)} = s_t^c , Y_{j,t}^{(c)} \geq Y_{j,t}^{(d)} \geq 0 }
        $$
        Where $\phi$ is the PMF of the Dirichlet-Multinomial distribution:
        $$
        \phi(\bm{Y}_t^{(c)}; s_t^{(c)}, \alpha \bm{X}_{t-1}^{(c)}) = \frac{\Gamma(\alpha s_{t}^{(c)}) \Gamma(s_{t}^{(c)} + 1)}{\Gamma((\alpha + 1) s_{t}^{(c)})} \prod_{j=1}^J \frac{\Gamma(Y_{j,t}^{(c)} + \alpha X_{j,t-1}^{(c)})}{\Gamma(Y_{j,t}^{(c)}) \Gamma(\alpha X_{j,t-1}^{(c)} + 1)}
        $$
    \end{enumerate}
    Note that, for computational ease, we draw samples from this distributions with Markov Chain Monte Carlo. However, exact methods, such as those based on rejection sampling \citep{awan2021privacy} or coupling from the past \citep{seeman2021exact}, can alleviate these computational issues.
\end{enumerate}

We synthesize 100 replicates of the entire synthetic data set for all time points and counties using each of the methods outlined above. We plot the synthesized case rate results as the boxplots in Figure \ref{fig:wass_base_complete}, with the solid lines representing the confidential (i.e., true) case rates. We single out three counties with different population sizes: Allegheny county (home to Pittsburgh, a large city), Erie county (a mid-sized suburban county), and Cameron couty (a small rural county). In the low-privacy regime, $\epsilon=1$, we see all the methods perform comparably well across county sizes, in that the boxplots of sanitized values indicate concentration around the non-private statistics. Note that this case was intentionally chosen to demonstrate how both our method and post-processing can preserve the most important structural features of the data in the low-privacy regime. However, in the high-privacy regime, $\epsilon = .01$, our method produces results conforming to the prior, whereas post-processing introduces noise that entirely degrades the statistical signal. This indicates that incorporating public information into the prior yields results with more effective weighting between prior public information and the newly introduced noisy statistical measurements. 

Furthermore, in Figure \ref{fig:wass_base_relative}, we plot the relative percentage error for COVID-19 case rates across privacy budgets and synthesis mechanisms for all counties and months in our study. As expected, the methods perform comparably well in the low-privacy regime ($\epsilon = 1$). However, in the high-privacy regime ($\epsilon = .01$), the counties and months with the largest relative error (i.e., the least populous counties) have smaller errors. Therefore our method helps to control the worst-case errors introduced by privacy preservation for analyzing small sub-populations, a frequent concern amongst social scientists \citep{santos2020differential}. 

As with any privacy mechanism operating on different statistics with different signal-to-noise ratios, different mechanism implementations will perform better or worse depending on the chosen data utility metric. We intentionally chose a low-privacy and a high-privacy regime to demonstrate where we expect to see improvements and where we expect post-processing based results to be similar. Our goal is not to reverse engineer metrics that unfairly make our mechanism look better. However, improved data utility from the mechanism is only one small goal of our work. Aside from the discussion of Figures \ref{fig:wass_base_complete} and \ref{fig:wass_base_relative}, we want to highlight one additional structural advantage our method has in comparison to the post-processing method. In order to make statistically valid inferences from these statistics, we need to rely on algorithms that evaluate the density of $Y$ given $X, Z$ up to a constant. This is neither analytically or computationally tractable under the post-processing method, but our methods accommodate the application of such measurement error corrections. As a result, our method offers a practical, operational advantage over post-processing that cannot be captured by a direct comparison of errors. 

\begin{figure}
    \centering
    \includegraphics[width=.9\textwidth]{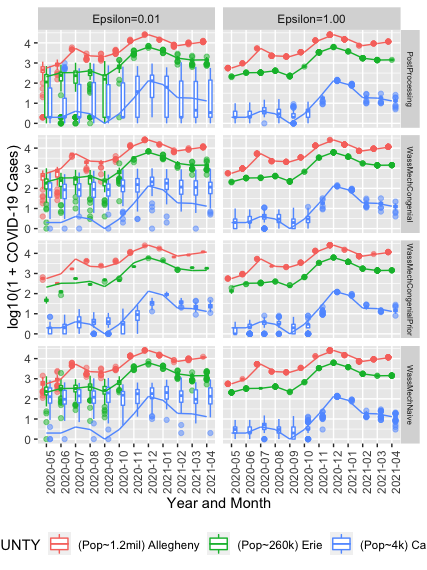}
    \caption{Simulated distribution of PPD replicates for COVID-19 case rates under different privacy budgets ($\epsilon=\{0.01, 1\}$) and synthesis strategies. Results shown for three counties in PA of different population sizes (Allegheny, 1.2m, red; Erie, 260k,green; Cameron, 4k, blue). Confidential (i.e., true) counts represented by solid line.}
    \label{fig:wass_base_complete}
\end{figure}

\begin{figure}
    \centering
    \includegraphics[width=.9\textwidth]{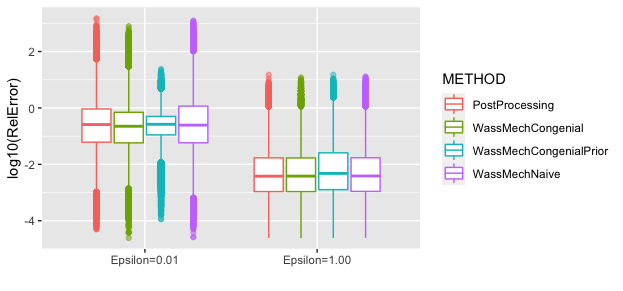}
    \caption{Simulated distribution of PPD relative error for COVID-19 case rates under different privacy budgets ($\epsilon=\{0.01, 1\}$) and synthesis strategies.}
    \label{fig:wass_base_relative}
\end{figure}

\subsection{Inference Analysis}

Algorithm \ref{alg:ppi-rej} was used in \citep{Seeman2020} to perform private posterior inference for mortality rates from historical CDC data. Here, we would like to compare the inferential improvements of results directly. However, this is intractable for one primary reason: we cannot easily compute the Bayes estimator $\delta_{Y^*}$ for our data, nor do we have an exact parameter estimate with which to compare. To get around this issue, in the proof of Theorem \ref{thm:expfamdom}, we used the fact that post-processing maps multiple values of $Y$ with different inferential interpretations to the same value of $Y^*$. Therefore we can investigate how frequently this effect occurs at different privacy budget configurations. This extends the analysis of \cite{santos2020differential} by specifically looking at the effect post-processing has on health disparities in the context of small-area estimates for COVID-19 prevalence. 

In particular, the post-processing as defined earlier produces the following effects:
\begin{itemize}
    \item \texttt{ContractedCaseZeros}: cases where multiple potential imputations of the private COVID-19 case data are contracted to 0
    \item \texttt{ContractedDeathZeros}: cases where multiple potential imputations of the private COVID-19 death data are contracted to 0
    \item \texttt{ContractedRates}: cases where the constrait that COVID-19 cases is bounded below by COVID-19 deaths contracts imputations of the COVID-19 survival rate to 0.
\end{itemize}

\begin{figure}[!htpb]
    \centering
    \includegraphics[width=.9\textwidth]{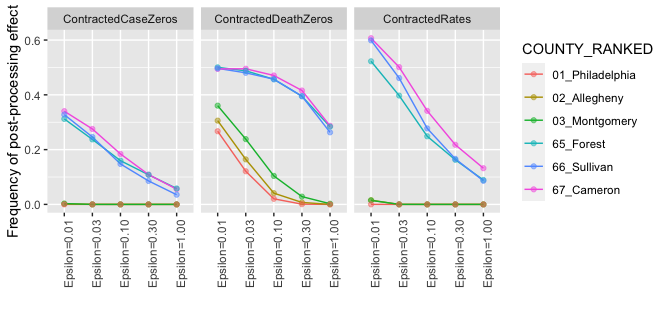}
    \caption{Frequency of post-processing effects for most and least populous counties in PA by privacy budget.}
    \label{fig:pp_inference_effects}
\end{figure}

We calculate the prevalence of these effects in our synthesis and plot the results in Figure \ref{fig:pp_inference_effects}. We only show the results for the cases with the largest and smallest numbers of total COVID-19 cases. In general, we see that degradation of inference due to post-processing occurs most frequently in the high-privacy, low-sample-size regime. These effects disproportionately fall on the least populous counties, demonstrating that post-processing exacerbates the inequitable inference on larger counties versus smaller counties present in smaller data sets. This problem is tempered not only by using the Wasserstein exponential mechanism, but also by performing regularization with a well-chosen base measure. 

\subsection{Distribution Specification and $\Delta_Z$}
\label{apx:delta_z}

In the previous subsections, we assumed a fixed $\Delta_Z = 2$ for our calculations, since our goal was to compare mechanisms at the same fixed Wasserstein Exponential Mechanism sensitivity. However, we could consider different data generating distributions to see how $\Delta_Z$ varies based on the calculations from Equation \label{eq:discrete_delta_z}. 

For simplicity, we'll consider one count from a database of size $n$. We consider a family of distributions $\Theta_{(z_0, z_1)}$, indexed by two parameters $z_1$ and $z_0$. Let $i \neq j$ and assume $Z$ defines the following probabilistic public information for all $\theta_{(z_0, z_1)} \in \Theta_{(z_0, z_1)}$:
$$
\p(X_i = 1 \mid s_{j1}, Z = z) \leq z_1 \in [0, 1], \quad \p(X_i = 1 \mid s_{j0}, Z = z) \geq z_0 \in [0, 1]
$$
In words, $z_1$ and $z_0$ capture the worst-case dependence between knowing one member of a secret pair (i.e., whether one person's binary secret value is 0 or 1) and the values of all the other counts that potentially agree with the secret. Using Equation \ref{eq:discrete_delta_z}, we numerically evaluate $\Delta_z$ at different choices for $n$, $z_1$, and $z_0$. We plot these results in Figure \ref{fig:count_deltaz_example}. As the dependence becomes more extreme, i.e. as $z_0$ approaches $0$ and $z_1$ approaches 1, $\Delta_z$ approaches $n$.

\begin{figure}
    \centering
    \includegraphics[width=.95\textwidth]{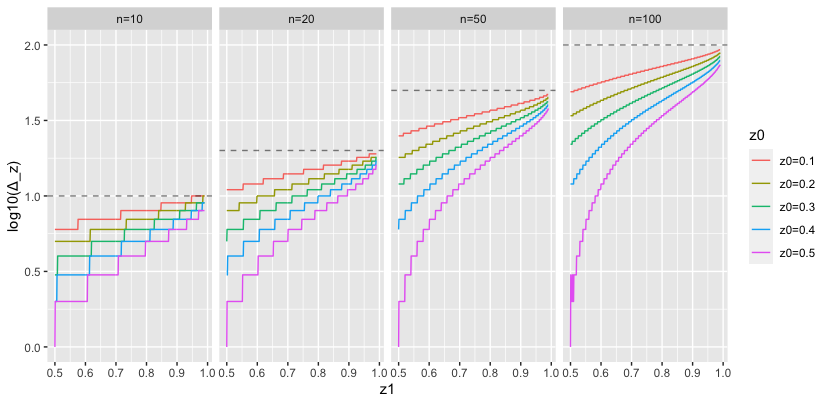}
    \caption{$\Delta_z$ of correlated counts under different distributional assumptions}
    \label{fig:count_deltaz_example}
\end{figure}

\section{Discussion}

Our work provides a statistically reasoned framework for understanding privacy guarantees for PPD by extending Pufferfish to consider distributions conditioned on public information. Although we have focused on $\epsilon$-DP style privacy semantics, we could have easily considered $f$-DP \citep{dong2019gaussian}, $\rho$-zero concentrated DP \citep{bun2016concentrated}, or other semantics which allow for an inferential testing perspective. Future work should extend the proposed approach to those frameworks.

Many view Pufferfish as harder to operationalize than DP, which some may view as a practical limitation of our work. To this, we have a few comments. First, as a reminder, we are addressing a class of problems for which DP alone does not fully describe the formal privacy loss from the data curator's actions; comparing our work to pure DP relies on a false equivalence between the underlying problems at hand. This is an important conceptual trade-off: DP may be easier to interpret, but it can fail to capture the actions taken by data curators in these public information scenarios. Second, we consistently show throughout this paper that the tools used in DP mechanism design can be reused in our setting, and our framework provides a way to re-contextualize those guarantees and more helpfully use public information. Our work is not meant to replace DP, but to complement it. Third, and most importantly, our framework allows for more flexible semantic privacy interpretations than those offered by the usual semantic interpretations associated with DP. This offers data curators an opportunity to be more transparent about their semantic privacy guarantees in unideal but realistic privacy-preserving data processing systems. We argue this is an essential part of making data processing more transparent to end users, both for better end data utility and data curator accountability. 

Note that our work is designed as a way to accommodate public information into privacy loss analysis, but this itself is not a unilateral \textit{endorsement} of releasing public information to accommodate DP results. The  privacy affordances offered by DP are kept closest to their idealized guarantees with no public information released outside the confines of the framework. Whenever possible, data curators should attempt to release results with randomization to satisfy DP, and carefully consider when certain cases that do not logistically accommodate this randomization occur. 

Still, our proposed work is valuable because situations that necessitate public information arise frequently in the social sciences, in particular for the production of official statistics, administrative data, and survey data. The US Census Bureau's public information use case was mandated by public law \citep{94th_congress_of_the_united_states_of_america_pl_1975}, but other cases, such as public disclosure of survey methodology \citep{seeman2021posterior}, help more generally to quantitatively establish trust in the quality of data created by the curator. In cases like these, disclosing public information not only satisfies an important logistical requirement, but also respects a context specific normative practice on how information flows between parties. This context-specificity is an essential component in how privacy is sociologically conceived \citep{nissenbaum2009privacy}, and DP alone often fails to align with this important sociological dimension of privacy. As a result, our framework serves to align important technical and sociological goals in privacy preserving data sharing and analysis currently faced by data curators. 

\acks{This work was supported by NSF SES-1853209. Seeman was partially supported by a U.S. Census Bureau Dissertation Fellowship. We declare no known conflicts of interest.}

\newpage

\appendix

\section{Proofs}
\label{apx:proofs}

\subsection{Proof of Lemma \ref{lem:etp-bayes}}

Let $z \in \mathcal{Z}$ $B\in \mathcal{F}_Y$, $\theta_z \in \mathbb{D}_{\mathrm{DPz}}$, and $(s_1, s_2) \in \mathbb{S}_{\mathrm{pairs,DPz}}$. By Bayes rule:
$$
\frac{\p(s_1 \mid Y \in B, \theta_z)/ \p(s_2 \mid Y \in B, \theta_z)}{\p(s_1 \mid \theta_z) / \p(s_2 \mid \theta_z)} = \frac{\p(Y \in B \mid s_1, \theta_z)}{\p(Y \in B \mid s_2, \theta_z)}
$$
Note that we avoid dividing by zero in either quantity because we assume $\theta_z \in \mathbb{D}_{\mathrm{DPz}}$ and $(s_1, s_2) \in \mathbb{S}_{\mathrm{pairs,DPz}}$. The ratio on the right is within $[e^{-\epsilon}, e^{\epsilon}$] if and only if Equation \ref{eq:puffcond} holds.
\subsection{Proof of Lemma \ref{lem:etp-pp}}
Fix $z \in \mathcal{Z}$, $(s_1, s_2) \in \mathbb{S}_{\mathrm{pairs,DPZ}}$, and $\theta_z \in \mathbb{D}_{\mathrm{DPz}}$. For $B \in \mathcal{F}_Y$, let $B^*_z \in \mathcal{F}_{Y^* \times Z}$ and define $B_z^{*-1} \triangleq \{ y \in \mathcal{Y} \mid h(y,z) = Y^* \}$. Then $B^{*-1}_z \in \mathcal{F}_{Y}$ for all $z \in \mathcal{Z}$ by measurability. This implies:
\begin{align}
\p(Y^* \in B^*_z \mid s_1, \theta_z) &= \p(Y \in B_z^{*-1} \mid s_1, \theta_z) \\
&\leq \p(Y \in B_z^{*-1} \mid s_2, \theta_z) e^\epsilon \\
&= \p(Y^* \in B^*_z \mid s_2, \theta_z) e^\epsilon .
\end{align}

\subsection{Proof of Lemma \ref{lem:etp-seqcomp}}
Let $Y \triangleq (Y_1, Y_2)$ on $(\mathcal{Y}, \mathcal{F}_Y)$. By the conditional independence assumption, there exists $B_1, B_2 \in \mathcal{F}_{Y_1} \times \mathcal{F}_{Y_2}$ such that:
\begin{align*}
\p(Y \in B \mid s_1, \theta_z) &= \p(Y_1 \in B_1 \mid s_1, \theta_z) \p(Y_2 \in B_2 \mid s_1, \theta_z) \\
&\leq \left( e^{\epsilon_1} \p(Y_1 \in B_1 \mid s_2, \theta_z) \right) \left( e^{\epsilon_2} \p(Y_2 \in B_2 \mid s_2, \theta_z) \right) \\
&= e^{(\epsilon_1 + \epsilon_2)} \p(Y \in B \mid s_2, \theta_z) .
\end{align*}

\subsection{Proof of Lemma \ref{lem:etp-parcomp}}
First, we consider $B = 2$. Let $Y \triangleq (Y_1, Y_2)$ on $(\mathcal{Y}, \mathcal{F}_Y)$. If $s_1$ and $s_2$ both describe entries
in the same partition (either $X_1$ or $X_2$, without loss of generality), then the result is trivially satisfied. If $s_1$ describes an entry in $X_1$ and $s_2$ describes an entry in $X_2$, then:
\begin{align*}
\p(Y \in B \mid s_1, \theta_z) &= \p(Y_1 \in Y_1^{-1}(B) \mid s_1, \theta_z) \\
&\leq  e^\epsilon \p(Y_1 \in Y_1^{-1}(B) \mid s_2, \theta_z) \\
&= e^\epsilon \p(Y \in B \mid s_2, \theta_z).
\end{align*}
Therefore in all cases, the result holds for $B = 2$. By induction, this holds for any sized finite partition, i.e. $B \in \mathbb{N}$.

\subsection{Proof of Lemma \ref{lem:etp-cond}}
This proof follows \citep[Theorem 2.1]{gong2020congenial} (in their setting with $k = 1$ and $\gamma = 1$). Without loss of generality, let $z = 1$ (otherwise, replace $\mathcal{X}^{*n}$ with $\mathcal{X}^n \setminus \mathcal{X}^{*n}$). Specifically, for $(s_1, s_2) \in \mathbb{S}_{\mathrm{pairs,DPz}}$ and $B \in \mathcal{F}_Y$:
\begin{align*}
\frac{\p(Y \in B \mid s_1, Y \in \mathcal{X}^{n^*})}{\p(Y \in B \mid s_2, Y \in \mathcal{X}^{n^*})} &= \frac{\p(Y \in B \cap \mathcal{X}^{n^*} \mid s_1)}{\p(Y \in B \cap \mathcal{X}^{n^*} \mid s_2)} \frac{\p(Y \in \mathcal{X}^{n^*} \mid s_2)}{\p(Y \in \mathcal{X}^{n^*} \mid s_1)} \\
&\leq \left( e^\epsilon \right) \left( e^\epsilon \right) \\
&\leq e^{2\epsilon}
\end{align*}

\subsection{Proof of Theorem \ref{thm:wassexp}}

Let $(s_1, s_2) \in \mathbb{S}_{\mathrm{pairs,DPz}}$ and $B \in \mathcal{F}_Y$. Let $\mu_{i,\theta_z} \triangleq \p(\cdot \mid \theta_z, s_i)$. Let $\gamma^*$ be the joint distribution that achieves the Wasserstein distance bound. Then:
\begin{align}
\frac{\p(Y \in B \mid s_1, \theta_z)}{\p(Y \in B \mid s_2, \theta_z)} &= \frac{\int_{\mathcal{X}} \p(Y \in B \mid s_1, \theta_z, X=x) \ d\mu_{1, \theta_z}(x) }{\int_{\mathcal{X}} \p(Y \in B \mid s_2, \theta_z, X=x) \ d\mu_{2, \theta_z}(x)} \\
&= \frac{\int_{\mathcal{Y}} \int_{\mathcal{X}} \ind{y \in B} \exp\left(-\frac{\epsilon L_x(y)}{2\sigma(\Delta_z)} \right) \ d\mu_{1, \theta_z}(x) \ d\nu_z(y) }{\int_{\mathcal{Y}} \int_{\mathcal{X}} \ind{y \in B} \exp\left(-\frac{\epsilon L_X(y)}{2\sigma(\Delta_z)} \right) \ d\mu_{2, \theta_z}(x) \ d\nu_z(y) } \\
&= \frac{\int_{\mathcal{Y}} \int_{\mathcal{X}} \int_{\mathcal{X}} \ind{y \in B}  \exp\left(-\frac{\epsilon L_x(y)}{2\sigma(\Delta_z)} \right) \ d\gamma^*(x, x') \ d\nu_z(y)}{ \int_{\mathcal{Y}} \int_{\mathcal{X}} \int_{\mathcal{X}} \ind{y \in B} \exp\left(-\frac{\epsilon L_{x'}(y)}{2\sigma(\Delta_z)} \right) \ d\gamma^*(x, x')  \ d\nu_z(y)}
\end{align}
Let $B_{\Delta_z}(x) \triangleq \{ x' \in \mathcal{X} \mid d(x,x') \leq \Delta_z \}$. Then by construction and definition of the Wasserstein distance:
\begin{align}
\frac{\p(Y \in B \mid s_1, \theta_z)}{\p(Y \in B \mid s_2, \theta_z)} &= \frac{ \int_{\mathcal{Y}} \int_{\mathcal{X}} \int_{B_{\Delta_z}(x)} \ind{y \in B}  \exp\left(-\frac{\epsilon L_x(y)}{2\sigma(\Delta_z)} \right) \ d\gamma^*(x, x') \ d\nu_z(y)}{ \int_{\mathcal{Y}} \int_{\mathcal{X}} \int_{B_{\Delta_z}(x')} \ind{y \in B} \exp\left(-\frac{\epsilon L_{x'}(y)}{2\sigma(\Delta_z)} \right) \ d\gamma^*(x, x') \ d\nu_z(y) }  \\
&\leq \frac{ \int_{\mathcal{Y}} \int_{\mathcal{X}} \int_{B_{\Delta_z}(x')} \ind{y \in B}  \exp\left(-\frac{\epsilon (L_{x'}(y) + \sigma(\Delta_z)) }{2\sigma(\Delta_z)} \right) \ d\gamma^*(x, x') \ d\nu_z(y)}{ \int_{\mathcal{Y}} \int_{\mathcal{X}} \int_{B_{\Delta_z}(x')} \ind{y \in B} \exp\left(-\frac{\epsilon L_{x'}(y)}{2\sigma(\Delta_z)} \right) \ d\gamma^*(x, x') \ d\nu_z(y) }  \\
&\leq \exp\left(\frac{\epsilon}{2} \right) \exp\left[\frac{\epsilon}{2} \left( \frac{ \sup \left\{ |L_x(y) - L_{x'}(y)| \mid x,x' \in \mathcal{X}^n, d(x, x') \leq \Delta_z  \right\} }{\sigma(\Delta_z)} \right) \right] \\
&\leq \exp(\epsilon)
\end{align}

\subsection{Proof of Theorem \ref{thm:wasskng}}

Let $L_X(y)$ be the loss function from the $K$-norm gradient mechanism, and let:
$$
L_X^*(y) \triangleq \frac{\norm{\nabla L_X(y)}_K}{\sigma_y(\Delta_Z)}
$$
Then in Theorem \ref{thm:wassexp} applied to $L_X^*(y)$, $\sigma(\Delta_z) = 1$. Therefore the Wasserstein $K$-norm gradient mechanism satisfies $\epsilon$-TP. 

\subsection{Proof of Theorem \ref{thm:wasskngutil}}

Here, we fully list the assumptions needed for this theorem:

\begin{enumerate}
    \item $n^{-1} \nabla^2 L_{X_n}(y)$ exists almost everywhere $\nu_{Z_n}$, and the smallest eigenvalue of $L_{X_n}(y)$ is greater than $\alpha > 0$ for all $n \in \mathbb{N}$, $y \in \mathcal{Y}$, $\nu_{Z_n}$ almost everywhere.
    \item There exists a unique $y^* \in \mathcal{Y}$ such that as $n \to \infty$: 
    $$
    Y^*_n \triangleq \argmin_{\tilde{y} \in \mathcal{Y}} L_{X_n}(\tilde{y}) \to_P y^*.
    $$
    \item For all $n \in \mathbb{N}$, $\Delta_{Z_n} \geq \Delta^* > 0$ w.p. 1
    \item $\sigma_y(\Delta^*)$ is continuous in $y$, and $\sigma_y(\Delta^*) \geq \sigma^*(\Delta^*)$ holds $\nu_{Z_n}$ almost everywhere.
    \item For all $n \in \mathbb{N}$
    $$
    \int_{\mathcal{X}} \exp\left(-\frac{\alpha \norm{y_n - y_n^*}_K}{2 \sigma_{y_n^*}(\Delta^*)} \right) d \nu_{Z_n} < \infty.
    $$
    \item $Z_n \to_D Z^*$ and $\nu_{Z_n} \to_D \nu_{Z^*} \ll \lambda$.
    \item Let $B_\delta(\cdot)$ be a $K$-norm ball of radius $\delta$. We assume that for $a \in \mathcal{Y}$ and some $\delta > 0$:
    $$
    \int_{B_\delta(y_n^*)} d\nu_{Z_n}(a)  = \Omega_P(1)
    $$
    In words, the base measure should support the true solution $y^*$ with probability bounded away from zero by some constant.
\end{enumerate}
Let $f_n(y)$ be the density of the Wasserstein $K$-norm gradient mechanism:
$$
f_n(y_n) \propto \exp\left(-\frac{\epsilon \norm{\nabla L_{X_n}(y_n)}_K }{2\sigma_y(\Delta_{Z_n})} \right) d\nu_{Z_n}(y_n),
$$
Let $a_n \triangleq n (y_n - y_n^*)$ with density $h_n(\cdot)$ so that:
$$
h_n(a_n) \propto \exp\left(-\frac{\epsilon \norm{\nabla L_{X_n}(y_n^* + a_n / n )}_K }{2\sigma_y(\Delta_{Z_n})} \right) d\nu_{Z_n}(y_n^* + a_n / n ).
$$
Note that in the transformation above, the Jacobian constant $n^{-1}$ gets absorbed in the proportionality. Next, we Taylor expand the loss function:
\begin{align}
\nabla L_{X_n}(y_n)  = \nabla L_{X_n}\left(y_n^* + a_n / n \right) = \nabla^2 L_{X_n}(y^*_n) \frac{a_n}{n} + o_p(1)
\end{align}
Following the proof of \citep[Theorem 3.2]{reimherr2019kng}, using Assumptions (1) - (4), there exists a constant $C > 0$ such that:
$$
-\frac{1}{\sigma_{y^*_n + a_n/n}(\Delta_{Z_n})} \norm{\nabla L_n(y^*_n + a_n/n) }_K \leq -\frac{C \alpha }{\sigma^*(\Delta^*)} \norm{a_n}_2.
$$
Then by the dominated convergence theorem, Assumptions (5) - (6), and the continuity of $\sigma$, we conclude that:
\begin{equation}
\label{eq:wass_kng_asymp_dens}
    \lim_{n \to \infty} h_n(a_n) = h(a) \propto \exp\left(-\frac{\epsilon}{2 \sigma_{y^*}(\Delta^*)} \norm{\Sigma^{-1} a}_K\right),
\end{equation}
with respect to the base measure $\nu_{Z^*}(a)$ (note that Assumption (6) ensures the asymptotic base measure supports $y^*$). This is a $K$-norm density with respect to the asymptotic base measure $\nu_{Z^*}$. Therefore, $Y_n$ concentrates around $Y_n^*$ at a rate at least as fast as for the standard $K$-norm gradient mechanism, i.e., $a_n = O_p(1)$, yielding the desired result.

\subsection{Proof of Lemma \ref{lem:etp-rep}}

Fix $\delta > 0$, $\bm{x} \in \mathcal{X}$, and let $d_H$ be a metric on $\mathcal{X}$. Then for sufficiently large $m$: 
\begin{align*} 
&\sup_{\theta_z \in \Theta_Z} \sup_{(s_1, s_2) \in \mathbb{S}_{\mathrm{pairs,DPz}}} W_{\infty} \{ \p(\cdot, \mid \theta_z, s_1), \p(\cdot, \mid \theta_z, s_2) \} \\
&\geq \sup_{(s_1, s_2) \in \mathbb{S}_{\mathrm{pairs,DPz}}} W_{\infty} \{ \p(\cdot, \mid \theta_{z,\bm{x}}^{(m)}, s_1), \p(\cdot, \mid \theta_{z,\bm{x}}^{(m)}, s_2) \} \\
&\geq 1 - \delta
\end{align*}
Since this is true for arbitrary $\delta$, then we can ensure $\Delta_z \geq 1$, yielding the desired result.

\subsection{Proof of Theorem \ref{thm:expfamdom}}

First we show that $Y$ has a monotone likelihood in $\theta$. Let $\Xi \triangleq Y - T(X)$. Then the marginal density of $Y$, $m_\theta(y)$ can be expressed as a function of the density of $\Xi$, $g_\Xi(\xi)$ as a convolution:
$$
g_\Xi(\xi) \propto \exp\left(-\frac{\epsilon}{2\sigma(\Delta_z)} L(\xi) \right) \implies m_\theta(y) \triangleq \int_{\mathcal{X}} g_\Xi(y - x) f_\theta(x) \ dx.
$$
Let $y_1, y_2 \in \mathcal{Y}$ such that $y_1 < y_2$ and $\theta_1, \theta_2 \in \Theta$ such that $\theta_1 < \theta_2$. Suppose first that $n = 1$. Note that $f_\theta(x)$ has a monotone likelihood because $f_\theta$ belongs to an exponential family. Similarly, $g_\Xi$ has a monotone likelihood ratio by construction. Then using the main result from \citep{wijsman1985useful}:
\begin{align*}
&\int_\mathcal{X} g_\Xi(y_1 - x) f_{\theta_1}(x) \ dx \int_\mathcal{X} g_\Xi(y_2 - x) f_{\theta_2}(x) \ dx \geq \int_\mathcal{X} g_\Xi(y_2 - x) f_{\theta_1}(x) \ dx \int_\mathcal{X} g_\Xi(y_2 - x) f_{\theta_1}(x) \ dx \\
&\implies m_{\theta_1}(y_1) m_{\theta_2}(y_2) \geq m_{\theta_2}(y_1) m_{\theta_1}(y_2) \\
&\implies \frac{m_{\theta_2}(y_1)}{m_{\theta_1}(y_1)} \leq \frac{m_{\theta_2}(y_2)}{m_{\theta_1}(y_2)}
\end{align*}

The results extends to arbitrary $n$ by the iid assumption, and therefore $Y$ has a monotone likelihood ratio in $\theta$. 

\medskip

Next, let $\phi: \mathcal{Y}^* \mapsto [0, 1]$ be an unbiased test for $H_0: \theta \leq \theta_0$ versus $H_1: \theta > \theta_0$. Fix the type I error as $\alpha$. Using the fact that $Y$ has a monotone likelihood ratio in $\theta$ and the unbiasedness of $\phi$, we know $\beta_\phi(\theta_0) \leq \beta_\phi(\theta_1)$ for all $\theta_0, \theta_1 \in \Theta$ such that $\theta_0 \leq \theta_1$, where:
$$
\beta_{Y^*}(\theta) \triangleq \E_\theta[\phi(Y^*)] = \E_\theta[\phi(Y^*) (\ind{Y = Y^*} + \ind{Y < Y^*} + \ind{Y > Y^*})].
$$
When $\theta = \theta_0$ and $Y < Y^*$, using the MLR property, there exists a test $\phi': \mathcal{Y} \times \mathcal{Z} \mapsto [0, 1]$ such that $\phi'(Y,Z) = \phi(Y^*)$ whenever $Y = Y^*$, and:
$$
\E_{\theta_0}[\phi'(Y,Z) \mid Y < Y^*]  \leq \E_{\theta_0}[\phi(Y^*) \mid Y < Y^*]
$$
Let $\gamma \triangleq \E_{\theta_0}[\phi(Y^*) - \phi'(Y,Z)] \in [0, \alpha]$. We can modify $\phi'$ such that, for $Y > Y^*$:
$$
\E_{\theta_0}[\phi'(Y,Z) - \phi(Y^*) \mid Y > Y^*] = \gamma 
$$
This yields a final test $\phi'(Y,Z)$ which dominates $\phi(Y^*)$ and is still level $\alpha$. As long as $\p_\theta(Y \neq Y^*) > 0$, the dominance is non-trivial; the new test may have a smaller type I error, type II error, or both. 

\subsection{Proof of Theorem \ref{thm:ppi-rej-valid}}

Throughout this proof, we take $\pi$ to refer to the density of the variables in question. The target density has the form:
\begin{align*}
\pi(\theta \mid y, z) &= \frac{\pi(\theta, y \mid z) }{\pi(y \mid z)} \\
&= \frac{\int \pi(\theta, x, y \mid z) \ d \nu_X }{\int \pi(y, x \mid z) \ d \nu_X} \\
&= \frac{\int \pi(\theta \mid z) \pi(x \mid \theta, z) \pi(y \mid x, z) \ d \nu_X }{\int \int \pi(\theta \mid z) \pi(x \mid \theta, z) \pi(y \mid x, z) \ d \nu_X \ d \nu_\Theta}
\end{align*}
Following \citep{gong2019exact}, let $A$ be a Bernoulli variable with success probability associated with accepting one proposal in the algorithm with public information $z$ and proposed samples $\theta^* \in \Theta$ and $x^* \in \mcX^n$. Define:
$$
C \triangleq \sup_{y^* \in \mathcal{Y}} \pi(y^* \mid x^*, z)
$$
Then:
\begin{align*}
\pi(\theta^* \mid A = 1, z) &= \int \frac{ \pi(\theta^*, x^*, A=1) \ d\nu_X}{\pi(A = 1)} \\
&= \frac{C \int \pi(\theta \mid z) \pi(x \mid \theta, z) \pi(y \mid x, z) \ d \nu_X }{C \int \int \pi(\theta \mid z) \pi(x \mid \theta, z) \pi(y \mid x, z) \ d \nu_X \ d \nu_\Theta} \\
&= \pi(\theta \mid y, z)
\end{align*}
Therefore accepted samples in Algorithm \ref{alg:ppi-rej} are exact draws from $\theta \mid Y, Z$.

\section{Complete Algorithms}
\label{apx:algs}

\begin{algorithm}[!htbp]
\SetAlgoLined
\KwResult{One sample from $\theta \mid Y, Z$ }
Sample $\theta^* \sim \pi(\theta \mid Z)$ \; 
Sample $X^* \sim \pi(\cdot \mid \theta^*, Z)$ \; 
Sample $U \sim \mathrm{Unif}(0,1)$ \;
\eIf{
$$
\frac{ \pi(Y \mid X^*, Z) }{ \sup_{y^* \in \mathcal{Y}} \pi(y^* \mid X^*, Z) } \leq U
$$
}{
Return $\theta^*$.
}{
Go to beginning.
}
\caption{Posterior rejection sampling conditional on public information}
 \label{alg:ppi-rej}
\end{algorithm}

\begin{algorithm}[!htbp]
\SetAlgoLined
\KwResult{Estimate of $\widehat{\E}[a(\theta) \mid Y, Z]$ using proposals from density $g$}
Sample $\theta^{(1)}, \cdots, \theta^{(m)} \sim \pi(\theta \mid Z=z)$ \; 
Sample $X^{(j)} \sim \pi(\cdot \mid \Theta=\theta^{(j)}, Z=z)$ \; 
Calculate:
$$
w^{(j)} = \frac{ \pi(Y \mid X^{(j)}, Z) \pi(\theta^{(j)} \mid Z)}{g(\theta^{(j)})}
$$ \\
Return:
$$
\widehat{\E}[a(\theta) \mid Y, Z] \triangleq \frac{\sum_{j=1}^m w^{(j)} a(\theta^{(j)})}{\sum_{j=1}^m w^{(j)}}
$$ 
 \caption{Posterior importance sampling conditional on public information}
 \label{alg:ppi-is}
\end{algorithm}

\section{Complete Example Derivations}
\label{apx:ex_derivs}

\begin{example}[Derivation of Example \ref{ex:mvn_cond}]
\label{apx:mvn_cond}
Suppose we have the setup as given in Equations \ref{eq:mean_expmech_puff} and \ref{eq:mean_expmech_joint}. First, let's consider the case where $Z \independent X_1, \overline{X}$, in which case we need only consider:
\begin{equation}
\label{eq:mean_expmech_joint_no_z}
\begin{pmatrix} \overline{X} \\ X_1 \end{pmatrix} \sim N\left( \begin{pmatrix} \mu \\ \mu \end{pmatrix}, \begin{pmatrix} \sigma^2/n & \sigma^2/n \\ \sigma^2/n & \sigma^2 \end{pmatrix}\right) \implies \overline{X} \mid X_1 = v \sim N\left(\mu + \frac{1}{n}(v - \mu), \frac{n-1}{n^2} \sigma^2 \right)
\end{equation}

Define $\p_{\theta,v_1}(B) \triangleq \p_\theta(\overline{X} \in B \mid X_1 = v_1)$. By definition:
$$
W_\infty\left(\p_{\theta, v_1}, \p_{\theta, v_2} \right) \triangleq \inf_{\gamma \in \Gamma(\p_{\theta, v_1}, \p_{\theta, v_2}) } \esssup_{(y_1, y_2) \in \Omega_{\gamma}} |y_1 - y_2| .
$$

The solution to the optimization problem is explicitly calculable. If $V_1 \sim \p_{\theta, v_1}$ and $V_2 \sim \p_{\theta, v_2}$, then the optimal transport solution takes the joint distribution where, w.p. 1:
\begin{equation}
\label{eq:mean_norm_opt_trans_sol}
V_2 = V_1 + \frac{v_2 - v_1}{n}.
\end{equation}
To give some intuition for Equation \ref{eq:mean_norm_opt_trans_sol}, if we fix $\theta$, then $\{ \p_{\theta,v} \mid v \in \mathbb{R} \}$ is a location family with a location parameter linear in $v$. Therefore in the $W_\infty$ optimal transport solution, each infinitesimal piece of probability mass travels the same constant distance as specified by the difference in conditional means.Any other solution would necessarily require a larger essential supremum.

\medskip

The optimal transport solution in Equation \ref{eq:mean_norm_opt_trans_sol} is independent of both $\mu$ and $\sigma$, and therefore only depends on the secret pairs $(s_1, s_2) \in \mathbb{S}_{\mathrm{pairs}}$ for which:
$$
\sup_{\theta \in \mathbb{D}} \sup_{(s_1, s_2) \in \mathbb{S}_{\mathrm{pairs}}} W_\infty(\p_{\theta,v_1}, \p_{\theta,v_2}) = \frac{\Delta}{n} 
$$
This allows us to instantiate the Wasserstein exponential mechanism. First, for all databases whose $L_1$ distance is bounded by $\Delta / n$ we have:
\begin{equation}
\sup_{X,X' \colon \norm{X - X'}_{L_1} \leq \Delta/n} |L_X(y) - L_{X'}(y)| \leq \Delta / n 
\end{equation}
Therefore the Wasserstein exponential mechanism gives us exactly the same result as the $\epsilon$-DP exponential mechanism in Equation \ref{eq:mean_expmech_dp} without additional public information.

Next, we introduce $Z$. Performing a similar analysis using Equation \ref{eq:mean_expmech_joint} gives us:
\begin{equation}
\label{eq:mean_expmech_cond_z}
\overline{X} \mid X_1 = v, Z = z \sim N\left(\mu + \Sigma_{XVZ}^T \Sigma_{VZ}^{-1} \begin{pmatrix}
v - \mu \\ z - \mu_{Z}
\end{pmatrix}, \frac{\sigma^2}{n} - \Sigma_{XVZ}^T \Sigma_{VZ}^{-1} \Sigma_{XVZ} \right).
\end{equation}
Again, by properties of the conditional multivariate normal, the resulting family is still a location family parameterized by a (more complicated) linear change in $v$. However, this indicates that the dependence affects the Wasserstein distance through the nuisance parameter $\sigma^2$, which canceled in the previous case:
$$
\sup_{\theta \in \mathbb{D}} \sup_{(s_1, s_2) \in \mathbb{S}_{\mathrm{pairs}}} W_\infty(\p_{\theta,v_1}, \p_{\theta,v_2}) = \sup_{\theta \in \mathbb{D}} \left[ \Sigma_{XVZ}^T \Sigma_{VZ}^{-1} \begin{pmatrix}
\Delta \\ z - \mu_{Z}
\end{pmatrix} \right] . 
$$
Therefore in this situation, the ``sensitivity" of the Wasserstein exponential mechanism loss depends explicitly on what kinds of dependence we allow between the private and public information, i.e. how we choose $\mathbb{D}$. 
\end{example}

\begin{example}[Count Data Satisfying $\epsilon$-TP and SDL measures]
\label{apx:count_sdl}Let $X \in \{ 0, 1 \}^n$, and our goal is to implement the Wassserstein exponential mechanism targeting $S_n \triangleq \sum_{i=1}^n X_i$ with $L_X(y) = |y - S_n |$. The secret pairs are of the form:
$$
s_{ij} = \{ \omega \in \Omega \mid X_i(\omega) = j \}, \quad \mathbb{S}_{\mathrm{pairs}} \triangleq \left\{ (s_{i0}, s_{i1}) \mid i \in [n] \right\}.
$$
Possible distributions conditioned on the secrets take the form:
$$
\p\left( S_n = k \mid s_{i0}, Z=z \right) = f_{\theta_z,0}(k), \quad \p\left( S_n = k \mid s_{i1}, Z=z \right) = f_{\theta_z,1}(k).
$$
Therefore, we need to consider:
\begin{equation}
\label{eq:discrete_delta_z}
 \Delta_z = \sup_{\theta_z \in \Theta_z} W_\infty\left(\p_{\theta_z,0}, \p_{\theta_z,1} \right) = \sup_{\theta_z \in \Theta_z} \sup_{t \in \mathcal{T}} | F^-_{\theta_z,0}(t) - F^-_{\theta_z,1}(t) | \in \{ 1, \dots, n \} ,  
\end{equation}
where
$$
\mathcal{T} = \left\{ \p\left(S_n \leq k \mid s_{j0}, Z=z \right) \mid j \in \{ 1, \dots n\} \right\}.
$$
We consider public information $Z$ to be the property that releasing $(S_n, n - S_n)$ satisfies different Statistical Disclosure Limitation (SDL) measures, namely $k$-anonymity \citep{sweeney2002k} and $L_1$-distance $t$-closeness \citep{li2007t} to a public statistic $t^*$:
\begin{equation}
\label{eq:ex_count_sdl}
\begin{cases}
S_n \text{ satisfies $k$-anonymity} \implies S_n \in \{ k, k + 1, \dots, n - k \} \\
S_n \text{ satisfies $L_1$-$t$-closeness} \implies \left| \frac{S_n}{n} - t^* \right| \leq t .
\end{cases}
\end{equation}
Because both of these are descriptors of the secret pairs, we can jointly satisfy some SDL measures as well as $\epsilon$-DP. However, this comes at the expense of the guarantees holding for fewer secret pairs.
\end{example}

Example \ref{ex:count_sdl} demonstrates how $\epsilon$-TP changes the unit of privacy analysis in the presence of $Z$. While SDL methods typically consider an equivalence class of databases with the same SDL properties, and DP methods typically consider an entire database schema, our $\epsilon$-TP considers a partial database schema. This means $\epsilon$-TP only confers protections for those databases which plausibly agree with $Z$, but this partial schema could be constrained by the practical privacy guarantees afforded by SDL.

\newpage 

\bibliography{references}

\begin{thebibliography}{57}
\providecommand{\natexlab}[1]{#1}
\providecommand{\url}[1]{\texttt{#1}}
\expandafter\ifx\csname urlstyle\endcsname\relax
  \providecommand{\doi}[1]{doi: #1}\else
  \providecommand{\doi}{doi: \begingroup \urlstyle{rm}\Url}\fi

\bibitem[{94th Congress of the United States of
  America}(1975)]{94th_congress_of_the_united_states_of_america_pl_1975}
{94th Congress of the United States of America}.
\newblock {PL 94-171: Redistricting Data}, 1975.

\bibitem[Abowd et~al.(2019{\natexlab{a}})Abowd, Ashmead, Simson, Kifer,
  Leclerc, Machanavajjhala, and Sexton]{abowd2019census}
John Abowd, Robert Ashmead, Garfinkel Simson, Daniel Kifer, Philip Leclerc,
  Ashwin Machanavajjhala, and William Sexton.
\newblock {Census topdown: Differentially private data, incremental schemas,
  and consistency with public knowledge}.
\newblock \emph{US Census Bureau}, 2019{\natexlab{a}}.

\bibitem[Abowd et~al.(2019{\natexlab{b}})Abowd, Kifer, Moran, Ashmead, and
  Sexton]{Abowd2019}
John Abowd, Daniel Kifer, Brett Moran, Robert Ashmead, and William Sexton.
\newblock {Census TopDown Algorithm: Differentially Private Data, Incremental
  Schemas, and Consistency with Public Knowledge}.
\newblock 2019{\natexlab{b}}.

\bibitem[Ashmead et~al.(2019)Ashmead, Kifer, Leclerc, Machanavajjhala, and
  Sexton]{ashmead2019effective}
Robert Ashmead, Daniel Kifer, Philip Leclerc, Ashwin Machanavajjhala, and
  William Sexton.
\newblock {Effective Privacy After Adjusting for Invariants with Appli-cations
  to the 2020 Census}.
\newblock Technical report, Technical Report. US Census Bureau, 2019.

\bibitem[Asi and Duchi(2020)]{asi2020}
Hilal Asi and John~C Duchi.
\newblock {Instance-optimality in differential privacy via approximate inverse
  sensitivity mechanisms}.
\newblock In H~Larochelle, M~Ranzato, R~Hadsell, M~F Balcan, and H~Lin,
  editors, \emph{Advances in Neural Information Processing Systems}, volume~33,
  pages 14106--14117. Curran Associates, Inc., 2020.
\newblock URL
  \url{https://proceedings.neurips.cc/paper/2020/file/a267f936e54d7c10a2bb70dbe6ad7a89-Paper.pdf}.

\bibitem[Awan and Rao(2021)]{awan2021privacy}
Jordan Awan and Vinayak Rao.
\newblock {Privacy-Aware Rejection Sampling}.
\newblock \emph{arXiv preprint arXiv:2108.00965}, 2021.

\bibitem[Awan and Slavkovi{\'{c}}(2018)]{awan2018differentially}
Jordan Awan and Aleksandra Slavkovi{\'{c}}.
\newblock {Differentially private uniformly most powerful tests for binomial
  data}.
\newblock \emph{Advances in Neural Information Processing Systems},
  31:\penalty0 4208--4218, 2018.

\bibitem[Barak et~al.(2007)Barak, Chaudhuri, Dwork, Kale, McSherry, and
  Talwar]{barak2007privacy}
Boaz Barak, Kamalika Chaudhuri, Cynthia Dwork, Satyen Kale, Frank McSherry, and
  Kunal Talwar.
\newblock {Privacy, accuracy, and consistency too: a holistic solution to
  contingency table release}.
\newblock In \emph{Proceedings of the twenty-sixth ACM SIGMOD-SIGACT-SIGART
  symposium on Principles of database systems}, pages 273--282, 2007.

\bibitem[Bassily et~al.(2013)Bassily, Groce, Katz, and
  Smith]{bassily2013coupled}
Raef Bassily, Adam Groce, Jonathan Katz, and Adam Smith.
\newblock {Coupled-worlds privacy: Exploiting adversarial uncertainty in
  statistical data privacy}.
\newblock In \emph{2013 IEEE 54th Annual Symposium on Foundations of Computer
  Science}, pages 439--448. IEEE, 2013.

\bibitem[Bassily et~al.(2020)Bassily, Cheu, Moran, Nikolov, Ullman, and
  Wu]{bassily2020private}
Raef Bassily, Albert Cheu, Shay Moran, Aleksandar Nikolov, Jonathan Ullman, and
  Steven Wu.
\newblock {Private Query Release Assisted by Public Data}.
\newblock In Hal~Daumé III and Aarti Singh, editors, \emph{Proceedings of the
  37th International Conference on Machine Learning}, volume 119 of
  \emph{Proceedings of Machine Learning Research}, pages 695--703. PMLR, 2020.
\newblock URL \url{https://proceedings.mlr.press/v119/bassily20a.html}.

\bibitem[Biswas et~al.(2020)Biswas, Dong, Kamath, and
  Ullman]{biswas2020coinpress}
Sourav Biswas, Yihe Dong, Gautam Kamath, and Jonathan Ullman.
\newblock {CoinPress: Practical Private Mean and Covariance Estimation}.
\newblock \emph{Advances in Neural Information Processing Systems}, 33, 2020.

\bibitem[Blackwell(1953)]{blackwell1953equivalent}
David Blackwell.
\newblock {Equivalent comparisons of experiments}.
\newblock \emph{The annals of mathematical statistics}, pages 265--272, 1953.

\bibitem[Bun and Steinke(2016)]{bun2016concentrated}
Mark Bun and Thomas Steinke.
\newblock {Concentrated differential privacy: Simplifications, extensions, and
  lower bounds}.
\newblock In \emph{Theory of Cryptography Conference}, pages 635--658.
  Springer, 2016.

\bibitem[Canonne et~al.(2019)Canonne, Kamath, McMillan, Smith, and
  Ullman]{canonne2019structure}
Clément~L Canonne, Gautam Kamath, Audra McMillan, Adam Smith, and Jonathan
  Ullman.
\newblock {The structure of optimal private tests for simple hypotheses}.
\newblock In \emph{Proceedings of the 51st Annual ACM SIGACT Symposium on
  Theory of Computing}, pages 310--321, 2019.

\bibitem[Desfontaines and Pej{\'{o}}(2019)]{desfontaines2019sok}
Damien Desfontaines and Balázs Pej{\'{o}}.
\newblock {Sok: differential privacies}.
\newblock \emph{arXiv preprint arXiv:1906.01337}, 2019.

\bibitem[Ding et~al.(2011)Ding, Winslett, Han, and Li]{ding2011differentially}
Bolin Ding, Marianne Winslett, Jiawei Han, and Zhenhui Li.
\newblock {Differentially private data cubes: optimizing noise sources and
  consistency}.
\newblock In \emph{Proceedings of the 2011 ACM SIGMOD International Conference
  on Management of data}, pages 217--228, 2011.

\bibitem[Dong et~al.(2019)Dong, Roth, and Su]{dong2019gaussian}
Jinshuo Dong, Aaron Roth, and Weijie~J Su.
\newblock {Gaussian differential privacy}.
\newblock \emph{arXiv preprint arXiv:1905.02383}, 2019.

\bibitem[Dwork et~al.(2006{\natexlab{a}})Dwork, Kenthapadi, McSherry, Mironov,
  and Naor]{dwork2006our}
Cynthia Dwork, Krishnaram Kenthapadi, Frank McSherry, Ilya Mironov, and Moni
  Naor.
\newblock {Our data, ourselves: Privacy via distributed noise generation}.
\newblock In \emph{Annual International Conference on the Theory and
  Applications of Cryptographic Techniques}, pages 486--503. Springer,
  2006{\natexlab{a}}.

\bibitem[Dwork et~al.(2006{\natexlab{b}})Dwork, McSherry, Nissim, and
  Smith]{dwork2006calibrating}
Cynthia Dwork, Frank McSherry, Kobbi Nissim, and Adam Smith.
\newblock {Calibrating noise to sensitivity in private data analysis}.
\newblock In \emph{Theory of cryptography conference}, pages 265--284.
  Springer, 2006{\natexlab{b}}.

\bibitem[Dwork et~al.(2014)Dwork, Roth, and {others}]{dwork2014algorithmic}
Cynthia Dwork, Aaron Roth, and {others}.
\newblock {The algorithmic foundations of differential privacy.}
\newblock \emph{Foundations and Trends in Theoretical Computer Science},
  9\penalty0 (3-4):\penalty0 211--407, 2014.

\bibitem[Dwork et~al.(2017)Dwork, Smith, Steinke, and Ullman]{dwork2017exposed}
Cynthia Dwork, Adam Smith, Thomas Steinke, and Jonathan Ullman.
\newblock {Exposed! a survey of attacks on private data}.
\newblock \emph{Annual Review of Statistics and Its Application}, 4:\penalty0
  61--84, 2017.

\bibitem[Fearnhead and Prangle(2012)]{Fearnhead2012}
Paul Fearnhead and Dennis Prangle.
\newblock {Constructing summary statistics for approximate Bayesian computation
  : semi-automatic approximate Bayesian computation}.
\newblock \emph{Journal of the Royal Statistical Society. Series B (Statistical
  Methodology)}, 74\penalty0 (3):\penalty0 419--474, 2012.

\bibitem[Foulds et~al.(2016)Foulds, Geumlek, Welling, and
  Chaudhuri]{foulds2016theory}
James Foulds, Joseph Geumlek, Max Welling, and Kamalika Chaudhuri.
\newblock {On the theory and practice of privacy-preserving Bayesian data
  analysis}.
\newblock \emph{arXiv preprint arXiv:1603.07294}, 2016.

\bibitem[Gao et~al.(2021)Gao, Gong, and Yu]{gao2021subspace}
Jie Gao, Ruobin Gong, and Fang-Yi Yu.
\newblock {Subspace Differential Privacy}.
\newblock \emph{arXiv preprint arXiv:2108.11527}, 2021.

\bibitem[Ghosh et~al.(2012)Ghosh, Roughgarden, and Sundararajan]{Ghosh2012}
Arpita Ghosh, Tim Roughgarden, and Mukund Sundararajan.
\newblock {Universally Utility-Maximizing Privacy Mechanisms}.
\newblock \emph{SIAM Journal on Computing}, 41\penalty0 (6):\penalty0
  1673--1693, 2012.

\bibitem[Gong(2019)]{gong2019exact}
Ruobin Gong.
\newblock {Exact inference with approximate computation for differentially
  private data via perturbations}.
\newblock \emph{arXiv preprint arXiv:1909.12237}, 2019.

\bibitem[Gong(2020)]{gong2020transparent}
Ruobin Gong.
\newblock {Transparent privacy is principled privacy}.
\newblock \emph{arXiv preprint arXiv:2006.08522}, 2020.

\bibitem[Gong and Meng(2020)]{gong2020congenial}
Ruobin Gong and Xiao-Li Meng.
\newblock {Congenial Differential Privacy under Mandated Disclosure}.
\newblock In \emph{Proceedings of the 2020 ACM-IMS on Foundations of Data
  Science Conference}, pages 59--70, 2020.

\bibitem[Hay et~al.(2010)Hay, Rastogi, Miklau, and Suciu]{hay2010boosting}
Michael Hay, Vibhor Rastogi, Gerome Miklau, and Dan Suciu.
\newblock {Boosting the accuracy of differentially private histograms through
  consistency}.
\newblock \emph{Proceedings of the VLDB Endowment}, 3\penalty0 (1-2):\penalty0
  1021--1032, 2010.

\bibitem[He et~al.(2014)He, Machanavajjhala, and Ding]{he2014blowfish}
Xi~He, Ashwin Machanavajjhala, and Bolin Ding.
\newblock {Blowfish Privacy: Tuning Privacy-Utility Trade-Offs Using Policies}.
\newblock In \emph{Proceedings of the 2014 ACM SIGMOD International Conference
  on Management of Data}, SIGMOD '14, page 1447–1458, New York, NY, USA,
  2014. Association for Computing Machinery.
\newblock ISBN 9781450323765.
\newblock \doi{10.1145/2588555.2588581}.
\newblock URL \url{https://doi.org/10.1145/2588555.2588581}.

\bibitem[Kasiviswanathan and Smith(2014)]{kasiviswanathan2014semantics}
Shiva~P Kasiviswanathan and Adam Smith.
\newblock {On the'semantics' of differential privacy: A bayesian formulation}.
\newblock \emph{Journal of Privacy and Confidentiality}, 6\penalty0 (1), 2014.

\bibitem[Kasiviswanathan and Smith(2008)]{kasiviswanathan2008note}
Shiva~Prasad Kasiviswanathan and Adam Smith.
\newblock {A note on differential privacy: Defining resistance to arbitrary
  side information}.
\newblock \emph{CoRR abs/0803.3946}, 2008.

\bibitem[Kifer and Machanavajjhala(2011)]{kifer2011no}
Daniel Kifer and Ashwin Machanavajjhala.
\newblock {No free lunch in data privacy}.
\newblock In \emph{Proceedings of the 2011 ACM SIGMOD International Conference
  on Management of data}, pages 193--204, 2011.

\bibitem[Kifer and Machanavajjhala(2014)]{kifer2014pufferfish}
Daniel Kifer and Ashwin Machanavajjhala.
\newblock {Pufferfish: A framework for mathematical privacy definitions}.
\newblock \emph{ACM Transactions on Database Systems (TODS)}, 39\penalty0
  (1):\penalty0 1--36, 2014.

\bibitem[Li et~al.(2007)Li, Li, and Venkatasubramanian]{li2007t}
Ninghui Li, Tiancheng Li, and Suresh Venkatasubramanian.
\newblock {t-closeness: Privacy beyond k-anonymity and l-diversity}.
\newblock In \emph{2007 IEEE 23rd International Conference on Data
  Engineering}, pages 106--115. IEEE, 2007.

\bibitem[Liu et~al.(2016)Liu, Chakraborty, and Mittal]{liu2016dependence}
Changchang Liu, Supriyo Chakraborty, and Prateek Mittal.
\newblock {Dependence Makes You Vulnberable: Differential Privacy Under
  Dependent Tuples.}
\newblock 2016.

\bibitem[McKenna et~al.(2018)McKenna, Miklau, Hay, and
  Machanavajjhala]{mckenna2018optimizing}
Ryan McKenna, Gerome Miklau, Michael Hay, and Ashwin Machanavajjhala.
\newblock {Optimizing error of high-dimensional statistical queries under
  differential privacy}.
\newblock \emph{Proceedings of the VLDB Endowment}, 11\penalty0 (10):\penalty0
  1206--1219, 2018.

\bibitem[McSherry and Talwar(2007)]{mcsherry2007mechanism}
Frank McSherry and Kunal Talwar.
\newblock {Mechanism design via differential privacy}.
\newblock In \emph{48th Annual IEEE Symposium on Foundations of Computer
  Science (FOCS'07)}, pages 94--103. IEEE, 2007.

\bibitem[Nissenbaum(2009)]{nissenbaum2009privacy}
Helen Nissenbaum.
\newblock \emph{{Privacy in context}}.
\newblock Stanford University Press, 2009.

\bibitem[{Pennsylvania Department of Health}(2022)]{pahealth}
{Pennsylvania Department of Health}.
\newblock {COVID-19 in Pennsylvania}, 2022.
\newblock URL
  \url{https://www.health.pa.gov/topics/disease/coronavirus/Pages/Coronavirus.aspx}.

\bibitem[Quick(2021)]{quick2021generating}
Harrison Quick.
\newblock {Generating Poisson-distributed differentially private synthetic
  data}.
\newblock \emph{Journal of the Royal Statistical Society: Series A (Statistics
  in Society)}, 184\penalty0 (3):\penalty0 1093--1108, 2021.

\bibitem[Reimherr and Awan(2019)]{reimherr2019kng}
Matthew Reimherr and Jordan Awan.
\newblock {KNG: The k-norm gradient mechanism}.
\newblock \emph{Advances in Neural Information Processing Systems}, 32, 2019.

\bibitem[Ruggles et~al.(2022)Ruggles, Flood, Foster, Goeken, Pacas,
  Schouweiler, and Sobek]{ipums}
Steven Ruggles, Sarah Flood, Sophia Foster, Ronald Goeken, Jose Pacas, Megan
  Schouweiler, and Matthew Sobek.
\newblock {IPUMS USA: Version 10.0 [dataset]}, 2022.
\newblock URL \url{https://www.ipums.org/}.

\bibitem[Santos-Lozada et~al.(2020)Santos-Lozada, Howard, and
  Verdery]{santos2020differential}
Alexis~R Santos-Lozada, Jeffrey~T Howard, and Ashton~M Verdery.
\newblock {How differential privacy will affect our understanding of health
  disparities in the United States}.
\newblock \emph{Proceedings of the National Academy of Sciences}, 117\penalty0
  (24):\penalty0 13405--13412, 2020.

\bibitem[Seeman and Brummet(2021)]{seeman2021posterior}
Jeremy Seeman and Quentin Brummet.
\newblock {Posterior Risk and Utility from Private Synthetic Weighted Survey
  Data}.
\newblock In \emph{Presented at World Meeting of the International Society for
  Bayesian Analysis (ISBA)}, 2021.

\bibitem[Seeman et~al.(2020)Seeman, Slavkovi{\'{c}}, and Reimherr]{Seeman2020}
Jeremy Seeman, Aleksandra Slavkovi{\'{c}}, and Matthew Reimherr.
\newblock {Private Posterior Inference Consistent with Public Information: a
  Case Study in Small Area Estimation from Synthetic Census Data}.
\newblock \emph{Privacy in Statistical Databases}, 2020.

\bibitem[Seeman et~al.(2021)Seeman, Reimherr, and
  Slavkovi{\'{c}}]{seeman2021exact}
Jeremy Seeman, Matthew Reimherr, and Aleksandra Slavkovi{\'{c}}.
\newblock {Exact Privacy Guarantees for Markov Chain Implementations of the
  Exponential Mechanism with Artificial Atoms}.
\newblock \emph{Advances in Neural Information Processing Systemsa}, 2021.

\bibitem[Slavkovic and Seeman(2022)]{slavkovic2022statistical}
Aleksandra Slavkovic and Jeremy Seeman.
\newblock {Statistical Data Privacy: A Song of Privacy and Utility}, 2022.
\newblock URL \url{https://arxiv.org/abs/2205.03336}.

\bibitem[Slavkovi{\'{c}} and Lee(2010)]{slavkovic2010synthetic}
Aleksandra~B Slavkovi{\'{c}} and Juyoun Lee.
\newblock {Synthetic two-way contingency tables that preserve conditional
  frequencies}.
\newblock \emph{Statistical Methodology}, 7\penalty0 (3):\penalty0 225--239,
  2010.

\bibitem[Snoke and Slavkovi{\'{c}}(2018)]{snoke2018pmse}
Joshua Snoke and Aleksandra Slavkovi{\'{c}}.
\newblock {pMSE mechanism: differentially private synthetic data with maximal
  distributional similarity}.
\newblock In \emph{International Conference on Privacy in Statistical
  Databases}, pages 138--159. Springer, 2018.

\bibitem[Song et~al.(2017)Song, Wang, and Chaudhuri]{song2017pufferfish}
Shuang Song, Yizhen Wang, and Kamalika Chaudhuri.
\newblock {Pufferfish privacy mechanisms for correlated data}.
\newblock In \emph{Proceedings of the 2017 ACM International Conference on
  Management of Data}, pages 1291--1306, 2017.

\bibitem[Soto and Reimherr(2021)]{soto2021differential}
Carlos Soto and Matthew Reimherr.
\newblock {Differential Privacy over Riemannian Manifolds}.
\newblock \emph{Advances in Neural Information Processing Systems}, 2021.

\bibitem[Sweeney(2002)]{sweeney2002k}
Latanya Sweeney.
\newblock {k-anonymity: A model for protecting privacy}.
\newblock \emph{International Journal of Uncertainty, Fuzziness and
  Knowledge-Based Systems}, 10\penalty0 (05):\penalty0 557--570, 2002.

\bibitem[Torkzadehmahani et~al.(2019)Torkzadehmahani, Kairouz, and
  Paten]{torkzadehmahani2019dp}
Reihaneh Torkzadehmahani, Peter Kairouz, and Benedict Paten.
\newblock {Dp-cgan: Differentially private synthetic data and label
  generation}.
\newblock In \emph{Proceedings of the IEEE/CVF Conference on Computer Vision
  and Pattern Recognition Workshops}, page~0, 2019.

\bibitem[Wijsman(1985)]{wijsman1985useful}
Robert~A Wijsman.
\newblock {A useful inequality on ratios of integrals, with application to
  maximum likelihood estimation}.
\newblock \emph{Journal of the American Statistical Association}, 80\penalty0
  (390):\penalty0 472--475, 1985.

\bibitem[Xiao et~al.(2021)Xiao, Ding, Wang, Zhang, and
  Kifer]{xiao2021optimizing}
Yingtai Xiao, Zeyu Ding, Yuxin Wang, Danfeng Zhang, and Daniel Kifer.
\newblock {Optimizing fitness-for-use of differentially private linear
  queries}.
\newblock \emph{Proceedings of the VLDB Endowment}, 14\penalty0 (10), 2021.

\bibitem[Zhang et~al.(2020)Zhang, Ohrimenko, and Cummings]{zhang2022attribute}
Wanrong Zhang, Olga Ohrimenko, and Rachel Cummings.
\newblock {Attribute Privacy: Framework and Mechanisms}, 2020.
\newblock URL \url{https://arxiv.org/abs/2009.04013}.

\end{thebibliography}

\end{document}